\documentclass[sigconf,nonacm,screen,pdfa]{acmart}

\usepackage{graphicx}
\usepackage{subcaption}
\usepackage{multirow}
\usepackage{caption}

\usepackage{tikz}
\usetikzlibrary{arrows.meta, bending, chains}
\usepackage[T1]{fontenc}
\usepackage[utf8]{inputenc}

\usepackage{tabularx}

\usepackage[colorinlistoftodos]{todonotes}
\presetkeys{todonotes}{inline}{}

\usepackage[super]{nth}

\usepackage{listings}
\usepackage{xcolor}
\usepackage{verbatim}

\definecolor{codegreen}{rgb}{0,0.6,0}
\definecolor{codegray}{rgb}{0.5,0.5,0.5}
\definecolor{codepurple}{rgb}{0.58,0,0.82}
\definecolor{backcolour}{rgb}{0.95,0.95,0.92}

\newcommand{\sskip}{\vspace{5pt}}

\lstdefinestyle{mystyle}{
    backgroundcolor=\color{backcolour},   
    commentstyle=\color{codegreen},
    keywordstyle=\color{magenta},
    stringstyle=\color{codepurple},
    basicstyle=\ttfamily\footnotesize,
    breakatwhitespace=false,         
    breaklines=true,                 
    captionpos=b,                    
    keepspaces=true,                 
    numbers=left,                    
    numbersep=5pt,                  
    showspaces=false,                
    showstringspaces=false,
    showtabs=false,                  
    tabsize=2
}

\usepackage{soul}
\usepackage{amsmath}

\usepackage{paralist}

\usepackage[noend, linesnumbered]{algorithm2e}

\AtBeginDocument{%
  }

\setcopyright{acmlicensed}
\copyrightyear{2024}
\acmYear{2024}
\acmDOI{XXXXXXX.XXXXXXX}
\acmConference[Conference acronym 'XX]{Make sure to enter the correct
  conference title from your rights confirmation emai}{2024}{Woodstock, NY}
\acmISBN{978-1-4503-XXXX-X/18/06}

\setcopyright{none}
\acmConference[]{}{}{}
\copyrightyear{}
\acmYear{}
\acmDOI{}
\acmISBN{}

\begin{document}

\title{Optimizing Queries with Many-to-Many Joins}

\author{Hasara Kalumin}
\affiliation{%
  \institution{University of Maryland}
  \city{College Park}
  \country{USA}}
\email{hdesilva@umd.edu}

\author{Amol Deshpande}
\affiliation{%
  \institution{University of Maryland}
  \city{College Park}
  \country{USA}
}
\email{amol@umd.edu}


\begin{abstract}
\noindent
As database query processing techniques are being used to handle diverse workloads, a key emerging challenge is how to efficiently handle multi-way join queries containing multiple many-to-many joins. While uncommon in
traditional enterprise settings that have been the focus of much of the query optimization work to date, such queries are seen frequently in other contexts such as graph workloads. This has led to much work on developing join
algorithms for handling cyclic queries, on compressed (factorized) representations for more efficient storage of intermediate results, and on use of semi-joins or
predicate transfer to avoid generating large redundant intermediate results. In this
paper, we address a core query optimization problem in this context.  Specifically, 
we introduce an improved cost model that more accurately captures the cost of a query plan in such scenarios,
and we present several optimization algorithms for query optimization that incorporate these new cost functions. 
We present an extensive experimental evaluation, that compares the factorized representation approach with a full semi-join reduction approach as well as to an approach that uses
bitvectors to eliminate tuples early through sideways information passing.
We also present new analyses of robustness of these techniques to the choice of the join order, 
potentially eliminating the need for more complex query optimization and selectivity estimation techniques.

\end{abstract}



\maketitle

\section{Introduction}

The increasing use of database query processing techniques in handling diverse workloads, such as graph workloads and machine learning workloads, has
led to significant research into developing new methods for query processing and optimization. A particular area of focus within this research is
the prevalence of many-to-many joins in these workloads. Such joins frequently generate large intermediate results, even when the final query result
is relatively small. This phenomenon has prompted extensive studies in recent years, including a large body of work on worst-case optimal join
algorithms~\cite{hung, ngo-wcoj, skew-wcoj, Veldhuizen2012LeapfrogTA, Freitag2020CombiningWO, AdoptingWcojs, remywang}, factorized or compressed intermediate representations~\cite{danolt, FDB, olt, AbulBasher2020AnswerGF, graphflowdb}, 
and semi-join or bloomfilter-based approaches to eliminate tuples
early~\cite{dynamicyannakakis,Tziavelis2019OptimalAF,TziavelisGR21,diamond,lip,bitvector_aware_query_proc,predicate_transfer}.

Although significant progress has been made in developing query processing techniques to minimize the generation of spurious intermediate results and
to generate appropriate intermediate representations, there has been comparatively less focus on the {\em query optimization} problem, i.e., finding
the optimal query execution plan given a query and relevant statistics (e.g., selectivity estimates). 
As the execution techniques themselves mature, this question is becoming more urgent as a pre-cursor to incorporating the techniques standard query processing engines.

In this paper, we address this problem in a systematic and end-to-end fashion by: (a) developing a new cost model that accounts for the effect of postponing generation of intermediate results and avoiding redundant probes into hash tables, and (b) investigating the query
optimization problem of finding the optimal {\em left-deep pipelined} query execution plan for an {\bf acyclic query} with many-to-many joins, for the three approaches.
Although the latter optimization problem is known to admit a simple solution called {\em rank ordering} with a standard cost function~\cite{kbz,eddies_joe,a-greedy, lip, learned-QO-jignesh}, 
as we show in this paper, this problem becomes much more intricate when the intermediate result generation is postponed.  
We present an approach to model the cost of
an query execution plan under standard uniformity and independence assumptions that are typically made in this literature. The optimization problem,
unfortunately, does not seem to admit a polynomial-time optimal solution in all cases, because of the inherent complex interactions between the join conditions. For 
the factorization and bitvector-based approaches, we develop an exponential time optimal algorithm as well as several natural heuristics, whereas for the semi-join
full-reduction approach, we develop a polynomial-time optimal algorithm. 
We note that the cost models we develop, and the
analysis of the optimization problem, are more broadly applicable (e.g., for optimizing queries with web services~\cite{Srivastava2006QueryOO} or expensive
predicates~\cite{joe-expensive-predicates}). 

\begin{figure}
    \centering
    \includegraphics[scale=0.55]{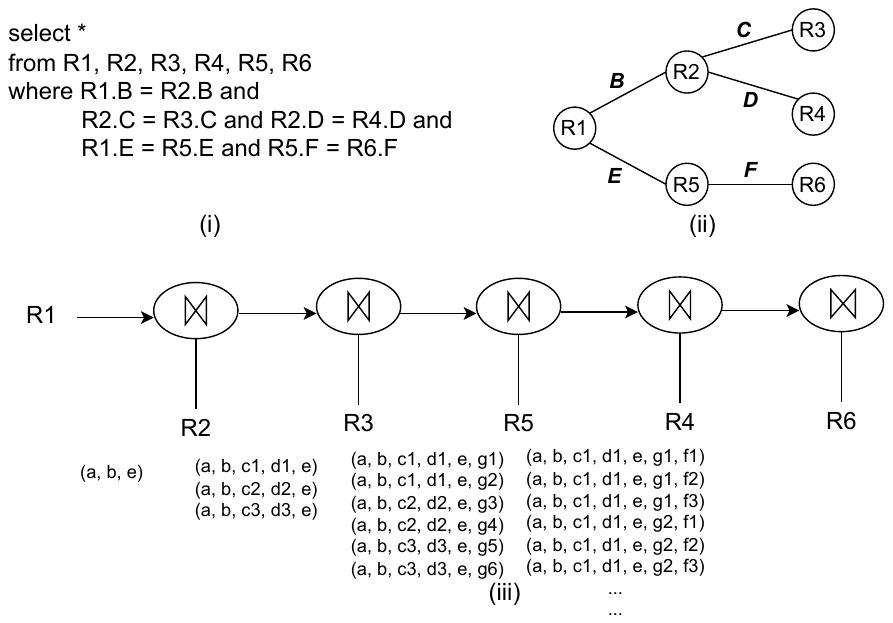}
    \caption{(i) An example 6-relation query used as the running example, and (ii) its join graph with edges annotated with the join attributes; (iii) A left-deep query plan and partial execution of a single $R1$ tuple.}
    \label{fig:running_example}
\end{figure}

We illustrate the key concepts and the motivation through an example 6-relation query shown in Figure \ref{fig:running_example}. One of the query
plans (that uses $R1$ as the ``driver'' relation) is shown as well as the sequence of tuples generated for a single $R1$ tuple $(a, b, e)$. Consider
the join with $R5$. In a typical implementation of such a left-deep pipelined plan, each of the 6 tuples shown in the example execution will be joined
with $R5$ (e.g., if this is a hash join, then each of the 6 tuples will be ``probed'' into the $R5$ hash table). However, we notice that all of these tuples must share the same value for the join attribute $E$;
this is because they were all generated from the same $R1$ tuple and the join with $R5$ is on an attribute of $R1$. So we can get away with just one probe into the hash table. We can similarly see the same issue with the join with $R4$. Here the first group of tuples that is shown (generated from $(a, b, c_1, d_1, e)$) have $D = d_1$, so only one probe into $R4$ is needed for these tuples. Although this is a somewhat exaggerated example, similar redundant probes can be seen in typical join workloads as we will show later. This inefficiency, sometimes called ``caching effect'', has been observed in previous literature~\cite{joe-expensive-predicates,aqp-paper,a-greedy}, but was not systematically explored until recently. 

A natural way to avoid such redundant probes is to maintain the intermediate tuples in a compressed or factorized representation (i.e., postpone some of the joins but keep enough information around to perform them later), so that the redundant probes are naturally avoided. This is the
approach we adopt, and has been considered in several recent works~\cite{danolt, graphflowdb, AbulBasher2020AnswerGF, diamond}. Another approach to
have a small cache in front of the join operator that is checked first~\cite{joe-expensive-predicates, stems}.  This is most useful when the probes
are expensive (e.g., if a probe involves a call to a web service or an API, or involves execution of an expensive UDF); however, in the case we
consider here, building such a cache dynamically adds significant overheads without any benefits (since a cache lookup is almost as 
expensive as the hash table lookup). An alternative approach is to use semi-joins or bitvector-based pruning to reduce the generation of spurious
intermediate results; this doesn't eliminate the redundant probes, but reduces the overall impact of the redundancy. For instance, assuming the tuple
$(a, b, e)$ contributes to the final result, it cannot be pruned away and will result in redundant probes shown in the example. The query optimization
problem needs to re-studied irrespective of the specific technique being used,
and our formalization can be applied to most such approaches.

A critical implication of accounting for this effect is that, once this is handled,
query processing becomes more {\bf robust} to the choice of the
join order; it still matters, but is less sensitive to estimation errors. The difference is especially stark for {\em star} queries where all the joins are with a single relation, and the
many-to-many joins almost don't matter from join order optimization perspective~\cite{bitvector_aware_query_proc,lip}. We highlight this because we have seen a number of recent works continue to use the
naive execution approach in presence of many-to-many joins. In that sense, our conclusions are similar to the recent work that showed that simple adaptations to the
query processing engines can reduce the need for more complex query optimization approaches~\cite{learned-QO-jignesh}.

We present a comprehensive experimental evaluation in a prototype system where we carefully implemented all of these techniques in a vectorized fashion, and illustrate 
the benefits of these techniques to drastically reduce, in presence of many-to-many joins, the 
overall execution cost as well as the number of probes into the hash tables, a more abstract cost metric. 

We begin with discussing the relevant background including the prior work on optimizing left-deep pipelined query plans and discuss the Yanakakis algorithm for 
minimizing the number of intermediate results created for an acyclic query. In Section~\ref{costmodel}, we present our cost model that properly accounts for the
postponement of generation of all intermediate results and discuss how that impacts the overall cost execution; we also present optimal algorithms as well 
as several fast heuristics to find best join orders. In Section~\ref{vectorized execution}, we describe our vectorized query execution engine. 
In Section \ref{evaluation}, we present our experimental evaluation, and put our work in context of other related work in Section \ref{related work}.

\section{Background}
\label{sec:background}

In this section, we discuss the two most closely related background topics namely, optimization of left-deep query plans and the Yannakakis algorithm
to minimize the number of intermediate results for acyclic queries. We discuss the other closely related work in Section \ref{related work}.

\subsection{Optimizing Left-Deep Pipelined Query Plans}
Let $Q$ denote a multi-way join query over relations $R_1, R_2, \cdots, R_n$ for which we are tasked to find an optimal execution plan. We assume $Q$ is
acyclic\footnote{The techniques we develop can be applied to cyclic queries in a standard fashion by choosing a spanning tree of the join
    graph~\cite{kbz}; however, the optimality results or the formal cost model don't directly generalize.}. 
    We also assume any selections are pushed down to the relations, and don't consider them explicitly any further (except in Section~\ref{selectivity estimation} 
            where we discuss selectivity estimation issues).
    We restrict our attention in this paper to ``left-deep'' (``left-linear'') query
execution plans, where the right input to each join must be one of the base relations. Figure \ref{fig:running_example} shows an example left-deep query plan, where $R_1$ (the
left input to the first join) is called the ``driver'' relation whose tuples are joined with each of the other relations one-by-one in some order.
Left-deep plans are attractive due to their pipelined nature, and are especially suitable in the context of {\em data streams} where the optimization issues
have been investigated in depth~\cite{eddies_joe,a-greedy,a-caching,flowalgorithms,skinnerdb}\footnote{We note that this class of plans is sometimes called
``right-deep''.}. In earlier work, such plans were executed ``tuple-at-a-time'', i.e., by taking a single tuple from 
the driver relation and propagating
it through all the join operators in the specified order, until results are produced or the tuple is thrown away due to not finding a match. In recent work
(including this work), the plans are executed ``batch-at-a-time'' to benefit from vectorization and to achieve better cache locality.

For the formal development, we assume a {\em generalized join} operator, which takes in a key (i.e., the value of the join attribute), and returns a collection
of ``matches'' from the corresponding relation. Although we primarily focus on hash joins, the formalism naturally captures index-nested-loops joins as well, with the cost
    of the probe (index lookup) modeled appropriately. It also allows us to generalize
to situations where external function calls may be made to find matches (e.g., to Web Services~\cite{Srivastava2006QueryOO} or LLMs or other types of API services), 
or expensive user-defined functions that return one or more outputs (e.g., using a {\em flatMap}-style operation or a {\em lateral join})~\cite{joe-expensive-predicates}. 
We call a single invocation of the join operator for a tuple a {\em probe}. In case of hash joins, a probe is equivalent to a hash table lookup, whereas for external API
calls, a probe is equivalent to one such call (in such scenarios, minimization of the number of probes becomes the key optimization metric given the monetary costs
associated with such external calls).

The join order problem here can, thus, be reduced to choosing the driver relation, and a permutation of the remaining relations. For simplicity, we assume
$R_1$ is always the driver relation in our examples as well as in the algorithms, and focus on finding the permutation of the remaining relations~\cite{kbz}. The
optimization algorithms can be ran once for each choice of the driver relation to find the overall optimal plan.

The {\em cost} of a given query plan $R_{i_1}, \cdots, R_{i_{n-1}}$ is estimated as: 
$$|R_1| (c_{i_1} + s_{i_1} * c_{i_2} + ...)$$
where $c_{i_1}$ denotes the cost of a single probe operation, and $s_{i_1}$ denotes the {\em selectivity} of the join operation (i.e., the expected number of matches per input tuple). Note that $s_i$ depends on the driver relation chosen, since that dictates the direction of the probe.
We assume the costs ($c_i$) are known in advance; in practice, the costs are likely to be different for different tuples, but we are not aware of any work in query
optimization that has modeled varying costs for a single join operator. 
Furthermore, the cost formula above makes an implicit assumption that the join operators are ``independent''. 

Given this cost model, the classical rank ordering algorithm~\cite{kbz,non-recursive-queries} sorts the join operators in the increasing order, while obeying precedence constraints, by:
$$(s_i - 1)/c_i$$ 
The precedence constraints are required to avoid cartesian products. For example, for the query shown in Figure \ref{fig:running_example}, only $\bowtie R_2$ and
$\bowtie R_5$ are eligible to be the first join operator (given $R_1$ is the driver relation), and if $R_2$ is chosen first, then $R_3$ and $R_4$ are also options for
the next join.
This simple algorithm is known to be optimal for the above cost model (which is very similar to what modern query optimizers use for left-deep pipelined plans). 
If correlations are present and can be modeled, the problem is known to be NP-Hard and admits constant-factor approximation algorithms~\cite{a-greedy}.

Although simple, this algorithm is very efficient and easy to analyze and has been widely used in the prior work in the data streaming and adaptive query
processing as well as recent work on robust query processing~\cite{skinnerdb,learned-QO-jignesh}.

However, a major problem with this cost model is the possibility of ``redundant'' probes in presence of many-to-many joins, as we discussed in the previous section. 
This observation has been made in several prior works. Earlier work on Eddies (specifically, SteMs~\cite{stems}) proposed using ``caching'' to avoid expensive probes, and the work on pipelined filter ordering~\cite{a-greedy}
proposed probing into join operators separately and doing a cross product at the end (they only discuss this for {\em star} queries). The recent line of work on {\em factorization}~\cite{danolt,FDB,olt} has investigated this and related problems more formally and
systematically. However, we are unaware of prior work on general cost models that work for arbitrary queries, or join order optimization when such techniques are used.

\subsection{Yannakakis Algorithm and Robust Query Processing}
For acyclic queries, the Yannakakis algorithm~\cite{Yannakakis1981AlgorithmsFA,dynamicyannakakis} guarantees an optimal total query execution time of $O(IN + OUT)$, where $IN$ denotes the total input size (across all relations)
and $OUT$ denotes the output size. The algorithm achieves this by ensuring that every intermediate tuple that is generated contributes to at least one output tuple.
This is done through two semi-join passes, whereby all tuples in the input relations that do not contribute to any output tuple are removed before doing the
actual joins. Although it is optimal, the two semi-join passes (and the additional pass at the end to generate results) makes this an expensive algorithm to use in practice 
(requiring a larger number of hash tables if hash joins are used,  including on the larger input relations which are often used as probe relations)\footnote{Note that, we use a
more efficient variation of this approach that requires fewer hash tables and is more commonly used in practice.}. Equally
importantly, this is a fundamentally ``blocking'' technique since no
output tuples can be produced until the two semi-join passes are completed. Figure \ref{fig:full_reduction} illustrates this approach for our running example, where we use the driver relation as the ``root''. 

\begin{figure}
    \centering
    \includegraphics[scale=0.6]{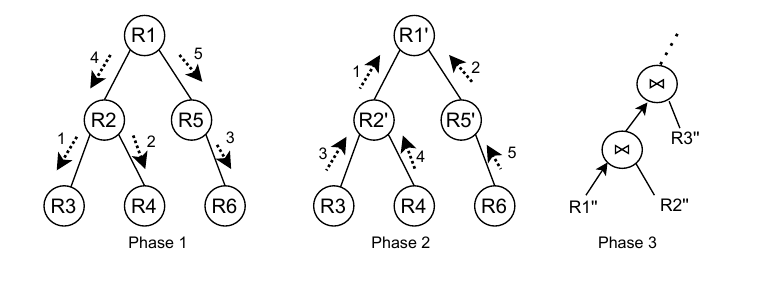}
    \caption{Yannakakis algorithm uses two semi-join passes to fully {\em reduce} the relations; arrows indicate the probe
        direction (e.g., Phase 1 first operation 1 is $R_2 \ltimes R_3$ to reduce $R_2$)}
    \label{fig:full_reduction}
\end{figure}

Practical adaptations of the above technique, typically through use of ``bloom filters'' or bitvectors, have been implemented in a number of systems over the years~\cite{lip, bitvector_aware_query_proc,predicate_transfer}. 
 
Figure \ref{fig:bitvector} shows an example where appropriate bitvectors are constructed (denoted by \textbf{BV}) and pushed down as far as possible.
Unlike what's shown in our running example (Figure~\ref{fig:running_example}), consider the scenario where $e$ finds no match in
$R_5$. The Yannakakis algorithm would have eliminated the tuple $(a, b, e)$ during its backward semi-join phase ($R_5' \ltimes R_1'$). 
In the bitvector-based plan shown in Figure~\ref{fig:bitvector}, $e$ will be checked against $BV(R_5.E)$ (and eliminated) before the join with $R_2$.
The requisite bloom filters or bitvectors can be constructed while building the hash tables.
False positives with bloom filters don't cause a correctness issue -- even if the tuples were generated, they would eventually be eliminated when joining with $R_5$.
Although effective in practice, this technique is only able to ``look-ahead'' a limited number of steps. For example, if those tuples were all eliminated by $R_4$ or
$R_6$ instead, we would still end up generating those intermediate tuples. The construction of the bloom filters and their use during execution are additional overheads
for this technique, but in our experimental evaluation, we found those overheads to be negligible. 

Both the above sets of techniques, as well as our approach, lead to increased ``robustness'' of the query execution engine, by reducing the sensitivity to the join
orders. This notion was formalized and analyzed in recent work on {\em look-ahead information passing}. We discuss this in more detail in Section~\ref{sec:robust}.

\begin{figure}
    \centering
    \includegraphics[scale=0.5]{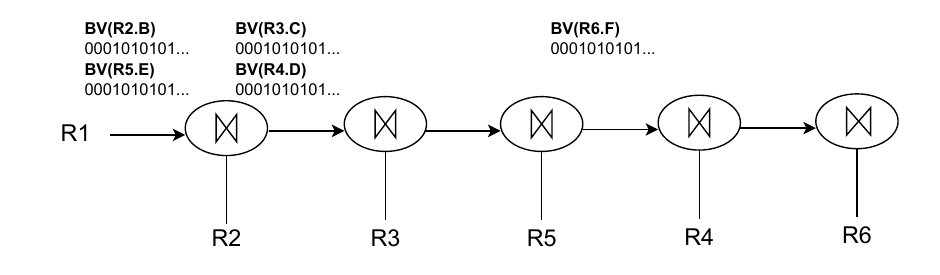}
    \caption{Pushing down bitvectors for early pruning}
    \label{fig:bitvector}
\end{figure}

\section{Cost Model and Optimization} \label{costmodel}
In this section, we formally develop the cost model that properly accounts for the avoidance of these redundant computations and discuss the intricacies of doing so. 
We then discuss how the cost model can be used to estimate the cost of a given plan under: (a) a factorized representation (COM), (b) a two-pass semi-join reduction (SJ), and (c) bitvector-based early pruning (BVP). Since COM is complementary to the other two, we consider four possible combinations, SJ+STD, BVP+STD, SJ+COM and BVP+COM, where STD denotes the standard execution and cost model that does not account for the use of factorized representation.
We then formulate the join order optimization problem, and develop several algorithms to solve it. We also analyze robustness of query plans under this approach following the recently proposed framework by~\cite{lip}.

\subsection{Selectivity $\rightarrow$ Match Probability and Fanout}
Much of the earlier work in this space does not make a distinction between the selectivity of a filter (predicate) and the selectivity of a join operator. 
However, those two are fundamentally different. The selectivity of a filter captures the probability that a random tuple from the input matches the filter. 
This can be extended to join operators to capture the average number of tuples produced a random input tuple (originally proposed by~\cite{kbz}); thus
the selectivity can be $> 1$, which is not an issue for the rank ordering algorithm. 
However, to model the number of probes accurately, we need to separate this into: 
\begin{itemize}
\item {\em Match probability ($m_i$),} the probability that an input tuple finds a match.
\item {\em Fanout ($fo_i$),} the average number of matches for an input tuple that does find a match.
\end{itemize}
It is easy to see that, for a join operator $\bowtie R_i$: $s_i = m_i \times fo_i$.

\subsection{Estimating Match Probabilities \& Fanouts}
\label{selectivity estimation}
Standard selectivity estimation techniques can be easily adapted to compute the match probabilities and fanouts. 
We briefly sketch adaptations of three selectivity estimation techniques here, and discuss how they can be modified to compute the match probabilities and fanouts with minimal additional statistics, and with the same computational cost as the original techniques.

For a join $R \bowtie_A S$, assuming $R$ is the probing relation, the simplest approach that assumes uniformity and
independence gives us: $s_i = \frac{1}{max(V(A,R), V(A, S))} |S|$, where $V(A, R)$ denotes the number of distinct values of $A$ in $R$. Under the same assumptions, we have that: 
$m_i = min(\frac{V(A,S)}{V(A,R)}, 1) = \frac{V(A, S)}{max(V(A, R), V(A, S))}$, and $fo_i = \frac{|S|}{V(A,S)}$.
If there is a predicate on $S$ with selectivity $s_p$, the estimate for $fo_i$ will be adjusted by that factor (unless $s_p|S| < V(A,S)$, in which case, we set $fo_i = 1$ and $m_i = min(\frac{s_p|A|}{V(A,R)}, 1)$). 

\begin{figure}
    \centering
    \includegraphics[width=3.0in]{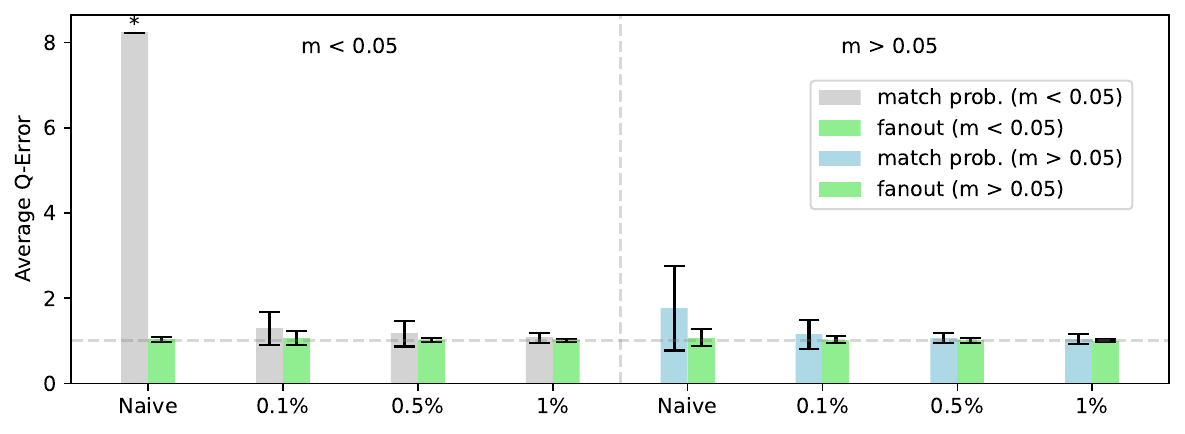}
    \caption{Sampling is highly effective at estimating match probabilities and fanouts ($* \rightarrow$ stddev. = 9.44)}
    \label{fig:expts_selectivity}
\end{figure}

Figure \ref{fig:expts_selectivity} shows the effectiveness of a sampling-based approach, which capture
correlations better, for estimating the match probabilities and fanouts. Here, we adapted the standard
correlated sampling approach~\cite{chen2017two,acharya1999join}, where we sample from one
relation uniformly at random, and with each sampled tuple, we maintain the number of matches it has in the other relation, as well as a uniform sample of its matches. The sample can be used to calculate match probabilities
and fanouts for queries of the type $\sigma_{R.a = x \wedge S.c = y}(R \bowtie_B S)$ (with appropriate scaling). In Figure \ref{fig:expts_selectivity}, we plot the Q-Error~\cite{moerkotte2009preventing} for the naive method discussed above,
as well as for three sample sizes; we used the DBLP datasets from the CE benchmark (Section \ref{evaluation}) where we generated queries randomly by choosing two of the (27) relations at random to join, and choosing the
predicates randomly. We separate out queries with low match probability since the error tend to be higher there. As we can see, even small samples are very effective at estimating these quantities, especially fanouts, whereas 
Naive performs poorly for queries with low match probabilities.

More sophisticated machine learning-based techniques~\cite{tzoumas,neurocard} 
typically already capture enough information to compute these components
separately. For instance, NeuroCard~\cite{neurocard} explicitly captures {\em join indicators} and {\em fanouts} as additional columns in the full outerjoin of all the relations in the database that it uses to learn a {\em deep
autoregressive neural network}; this is sufficient to compute $m_i$ and $fo_i$. As an example, if $1_R$ denotes the join indicator variable capturing whether a tuple from $R$ participates in the join
with $S$, then given a query: $\sigma_{R.B = 10}(R \bowtie_A S)$, the match probability can be seen as $P(R.B = 10 \wedge 1_R = 1)$ (we omit the fanout scaling corrections). 
We leave a more detailed exploration of such techniques to future work.

\subsection{Estimating the Cost of a Plan: COM}
\label{sec:costing}
The first question we look at is how to estimate the cost of a given plan assuming that we avoid redundant probes through use of a factorized intermediate representation; along the way we also formalize exactly what it means to avoid redundant probes. 
Following the approach developed for the rank ordering algorithm, we assume that every join operator has a constant cost per probe ($c_i$), and the main question here is to estimate the 
number of probes into each join operator. We first show this through an example and then develop the more general formulation.

Consider the join query shown in Figure \ref{fig:cost_model_figures}(i) with the driver relation $R_1$ and join order $R_2, R_3, R_5, R_4, R_6$. 
Letting $N$ denote $|R_1|$ (after any selections), the number of probes can be estimated as:
\begin{itemize}
\item $\bowtie R_2$: Since this is the first relation being probed, the number of probes = $N$.
\item $\bowtie R_3$: The join with $R_2$ produces $N * m_2 * fo_2 = N * s_2$ tuples, each with potentially a different value of $C$. Thus the number of probes here is $N*s_2$.

\item $\bowtie R_5$: The join with $R_5$ is on an attribute of $R_1$. To estimate the number of probes (which must be $< N$), we need to estimate the probability that a tuple of $R_1$ ``survives'' the two joins earlier. This
can be calculated to be\footnote{This is a simplification that assumes $E(c^Y) = c^{E(Y)}$ -- this is obviously not true in general, but is a reasonable approximation that we make to keep the exposition simple.}:
$$m_2 \times (1-(1 - m_3)^{fo_2})$$
The first term ($m_2$) captures the probability that the tuple finds a match in $R_2$. Assuming it found a match, $fo_2$ tuples are generated; the second term captures the probability that at least one of those $fo_2$ tuples
matches with $R_3$. 

\item $\bowtie R_4$: The number of $R_1$ tuples that survive joins with $R_2$ and $R_5$ is: $N * m_2 * m_5$, each of which produces $f_2$ tuples after join with $R_2$. Of those, $N * m_2 * m_5 * f_2 * m_3$ tuples also have a
match in $R_3$, and will need to be probed into $R_4$. We implicitly made use of the fact that the order in which previous joins are done doesn't matter, leading to a simplified calculation.

\item $\bowtie R_6$: Finally, the number of probes into $R_6$ can be estimated to be:
$$N \times m_2 \times (1-(1 - m_3m_4)^{fo_2}) \times m_5 \times fo_5$$
\end{itemize}

\begin{figure}
    \centering
    \includegraphics[width=3.7in]{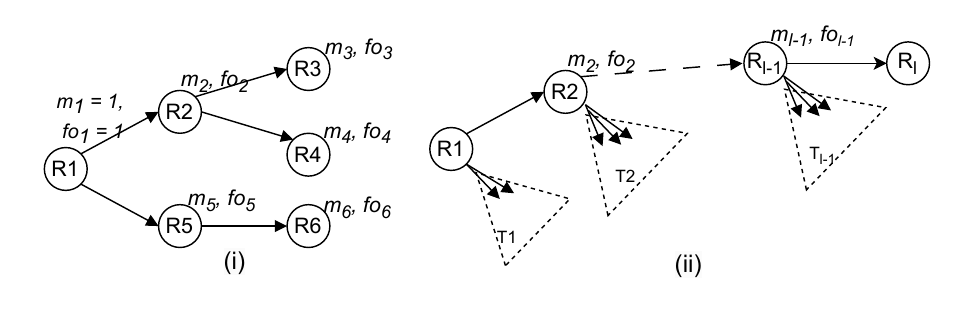}
    \caption{(i) Match probabilities and fanouts for the running example; (ii) A partially evaluated query plan}
    \label{fig:cost_model_figures}
\end{figure}

It would be useful for us to more succinctly capture the ``survival'' probability for a given set of join operators. Specifically, given a connected tree of join operators, $T$, consisting of node $T_r$ and 
a connected set of its descendants, 
we define $m_{T}$ to be the probability that a tuple survives all of those join operators. The list of join operators is sufficient to define $T$, so we will use that when convenient.

In general, $m_T$ can be recursively computed as follows. Let $T_1, \cdots, T_k$ denote the children of the root $T_r$ of $T$. Then: 
$$ m_T = m_{T_r} \times (1 - (1 - m_{T_1}m_{T_2}\cdots m_{T_k})^{fo_{T_r}})$$
To simplify the formulas, we assume $m_r$ (for the root node) $ = 1$.

As an example, $m_{1, 2, 3, 4} = m_1 \times m_2 \times (1 - (1 - m_3m_4)^{fo_2})$;
this allows us to more succinctly represent the overall cost of the above plan as:
$$N \times (1 + m_2fo_2 + m_{2,3} + m_2m_5fo_2m_3 + m_{1,2,3,4}m_5fo_5)$$
In contrast, without this optimization, we would have:
\begin{align*}
 N \times & \, (1 + m_{2}fo_{2} + m_{2}fo_{2}m_{3}fo_{3} + m_{2}fo_{2}m_{3}fo_{3}m_{5}fo_{5} +\\
& m_{2}fo_{2}m_{3}fo_{3}m_{5}fo_{5}m_{4}fo_{4})
\end{align*}
The two expressions are equivalent if $fo_? = 1$ for all the operators; however, as the fanouts increase, the difference between these terms becomes large as we show in our evaluation.

Next, we develop the more general formula. Consider a partially evaluated join order as shown in Figure \ref{fig:cost_model_figures}(ii), where our goal is to estimate the cost of the join operator $R_l$, where $R_2, .., R_{l-1}$ denote the ancestors of $R_l$. For each of the ancestors, we may have evaluated some of their other descendants, captured by $T_i$. Then, the estimated number of probes in $R_l$ is:
\begin{equation}
N \times m_2 fo_2 \times m_3 fo_3 \times ... m_{l-1} fo_{l-1} \times m_{T_1} \times m_{T_2} \times ...
\label{eq:probecost}
\end{equation}
Intuitively, we keep expanding (fanning out) along the path from root to $R_l$, but for any branches, we only care if the tuples survive the branch or not. If there no
branches, then this reduces to the more typical formula (Section \ref{sec:background}). We omit a detailed derivation of this formula due to space constraints, but we note that this is the expected number of probes assuming that, for a given join operator, any tuple either doesn't match or has exactly the same fanout. We also note that, as expected, this expression does not depend on the specific order in which all of those prior operators are evaluated.

\subsection{Join Order Optimization: COM}
The rank ordering algorithm is based on the {\em adjacent sequence interchange (ASI)} property that the
simpler cost function obeys~\cite{kbz}. Specifically, let $A$ and $B$ be two sequences (of join operators) and $V$ and $U$ be two non-empty sequences. We say that a cost function $C$ has the ASI property, if and only if there exists a function $T$ and a rank function defined as: $rank(S)= (T(S) - 1)/ C(S)$ 
such that the following holds: \\[2pt]
$C(AUVB) \le C(AVUB) \implies rank(U) \le rank(V)$, if $AUVB$ and $AVUB$ satisfy the precedence constraints.

Unfortunately, we can prove that:
\begin{theorem}
\label{thm:example}
The cost function developed above does not satisfy the ASI property.
\end{theorem}

\begin{proof}
We construct a counter-example with 6 joins, where $R_1$ (driver) joins with $R_2$ and $R_3$, that in turn join with $R_4, R_5$ and $R_6, R_7$ respectively. 
We set $m_i = 0.5$ for all $i$, and $fo_i = 1$ for all $i$ except $fo_2$ and $fo_3$. We consider two orders:  \\[2pt]
$R_2 \rightarrow R_3 \rightarrow R_4 \rightarrow R_7 \rightarrow R_5 \rightarrow R_6$ and \\[2pt]
$R_2 \rightarrow R_3 \rightarrow R_4 \rightarrow R_7 \rightarrow R_6 \rightarrow R_5$ \\[2pt] (i.e., $U = R_5$ and $V = R_6$). 

Because of the symmetry, no matter what $T(S)$ and $C(S)$ are, $rank(U) = rank(V)$. However, the costs of these two orders are
different if $fo_2 \ne fo_3$. In other words, the choice between these two orders depends on whether $fo_2 < fo_3$, which contradicts
the ASI property.
\end{proof}

We have implemented and evaluated an exhaustive optimal algorithm as well as two heuristics. 
The exhaustive algorithm (Algorithm 1) is a standard dynamic programming algorithm that finds the optimal order for every possible prefix of a valid join order (i.e., every connected subtree of the join graph that includes the driver relation) starting with prefixes of size 1. The cost function above does obey the principle of optimality, i.e., in the optimal join order, every prefix of it also has the optimal order by itself. This algorithm runs in time $O(n2^n)$ for a query with $n$ relations, but is much faster for non-star queries. We observe that the ASI property is obeyed in certain cases, either fully or partially (e.g., star queries obey it fully), and this can be used to further reduce the running time for the optimal algorithm, but we do not explore this further in this paper. 

\RestyleAlgo{ruled}
\SetKwComment{Comment}{/* }{ */}
\begin{algorithm}[t]

\SetAlgoLined
\KwIn{$\mathcal{J}$: rooted join tree with driver as the root ($r$); we use $\mathcal{J}$ to also denote the set of nodes in $\mathcal{J}$}
$CTs = \{S: S \subseteq \mathcal{J}, r \in S, S\ is\ connected\}$\;
\For{$T$ \textbf{in} $CTs$ \textbf{in} increasing order by size}{
    \For{$n \in T$} {
       $prefix = T \setminus \{n\}$\;
       \If{$prefix \in CTs$}{
          $cost = best[prefix] + est.\ probes\ into\ n\ (Eq. \ref{eq:probecost})$\;
          \If{$cost < best[T]$}{
            $best[T] = cost$
          }
       }
    }
}
$\textbf{return } best[\mathcal{J}]\ and\ the\ corresponding\ order$\;
\caption{Optimal Algorithm (given a driver)}
\end{algorithm}

We also experiment with three greedy heuristics each of which chooses the next operator greedily using the
formalism above, while obeying the precedence constraints, based on:
\begin{itemize}
    \item {\bf Selectivity ($m_i \times fo_i$)}: this simulates the rank ordering algorithm, which is equivalent to what today's query optimizers would use if restricted to left-deep plans.
    \item {\bf Expected number of tuples:} here we choose the next join operator such that the expected number of tuples after that join is minimized among all choices. This is analogous to a common query optimization heuristic where the next operator is chosen such that the intermediate result size after that operator is minimal among all options.
    \item {\bf Survival probability:} here we choose the next join operator such that the total survival probability of the prefix is minimized.
\end{itemize}

As we show through extensive experimental evaluation, the last heuristic finds plans nearly as good as the optimal algorithm in almost all cases. However, we can 
prove that:
\begin{theorem}
\label{thm:example1}
For any $f$ and any of the three heuristics above, there exists an input for which the heuristic produces a plan that is a factor $f$ worse than the optimal plan.
\end{theorem}

The worst cases are designed by hiding an operator with $m = 0$ under another operator with a high $fo$, so that the greedy heuristics don't consider that branch. 
Although the heuristics can be tweaked to avoid that worst case, other worst case inputs can be designed using similar ideas. 

\subsection{Costing and Optimization: BVP and COM+BVP}
Next, we discuss how our formalism can be used to estimate the cost of a given plan under bitvector-based early pruning (BVP) and the 
combination of BVP and COM.

One important difference here is that we need to count the number of probes into the bitvectors separately from the number of probes into the hash tables, since the former is significantly cheaper. We assume a false positive probability of $\epsilon$ for the bitvectors. Given that, the number of probes into the bitvectors and the hash tables can be estimated in a relatively straightforward manner. Consider the plan shown
in Figure \ref{fig:bitvector}. For each driver tuple (from $R_1$), the first probe is into the bitvector for $R_2$, which is successful with a probability of $m_2+\epsilon$ (since $m_2$ is the probability of the tuple finding a match in $R_2$). Hence, after the first two bitvectors are probed, we expect a total of $N \times (m_2+\epsilon)(m_5+\epsilon)$ tuples (from $R_1$) to remain. Those tuples will be probed into the hash table on $R_2$, and the resulting tuples will be probed into the bitvectors for $R_3$ and $R_4$ in that order. Overall, we get:
\begin{align*}
N \times &(1 + (m_2+\epsilon) + m_2(m_5+\epsilon)fo_2 + m_2(m_5+\epsilon)fo_2(m_3+\epsilon) +  \\ 
  &      m_{2}m_5fo_{2}m_{3}fo_{3}(m_{4}+\epsilon)fo_5)
\end{align*}
probes into the bitvectors; total probes into the hashtable can be calculated as:
\begin{align*}
 N \times & \,((m_{2}+\epsilon)(m_{5}+\epsilon) + m_2(m_{5}+\epsilon)fo_{2}(m_3+\epsilon)(m_4+\epsilon) + \\
&m_{2}(m_5+\epsilon)fo_{2}m_{3}(m_4 + \epsilon)fo_{3} + m_{2}m_5fo_{2}m_{3}(m4+\epsilon)\\
&fo_{3}fo_5(m_6+\epsilon) + m_{2}fo_{2}m_{3}fo_{3}m_{4}fo_{4}m_{5}fo_{5}(m_6+\epsilon))
\end{align*}

Although the above formulas look complex, the implementation only requires minor modifications to the cost models from the previous seciton. 
We can similarly obtain an estimate for the cost of a join order that avoids
the repeated probes by combining the procedure developed in the prior section with this. The main difference is in the probes into $R_5$. 
The above formula counts:\\[4pt]\mbox{\ \ } $N \times m2(m5+\epsilon)\underline{fo_2m_3(m_4+\epsilon)fo_3}$\\[4pt] probes into $R_5$. However, since the probes into $R_5$ are on an attribute of $R_1$,
we only care whether a tuple from $R_1$ survives the joins with $R_2$ and $R_3$ (along with surviving the bitvector probes into $R_5$ and $R_4$). 
This can be estimated as: \\[4pt] \mbox{\ \ }$N \times m_2(m_5+\epsilon)\underline{(1-(1-m_3(m_4+\epsilon))^{fo_2})}$.  \\[4pt]
As before, the main difference here is that the fanouts are taken out of equation since they don't matter for the probe into $R_5$.

Ding et al.~\cite{bitvector_aware_query_proc} noted that the principle of optimality does not hold when bitvectors are used for early pruning, and claim that the 
optimization time increases exponentially with the number of relations. However, when restricted to left-deep plans, this is not entirely correct. 
Assuming all possible bitvectors are used and pushed down as far as possible, we can show that: 
\begin{theorem}
\label{thm:optimality}
When restricted to left-deep plans and a specific driver relation, the principle of optimality holds for the cost model that uses bitvectors for early pruning. 
\end{theorem}

The key insight here is that, when the driver is fixed, the output from a sub-plan is independent of the join order used until that point. The example 
given by Ding et al.~\cite{bitvector_aware_query_proc} uses two left-deep plans that have different driver relations. 

Thus, the complexity of the optimization problem does increase, but {\bf linearly}, not exponentially. In other words, we need to find the optimal plan for each of the $n$ possibilities for the driver, and pick the best among those. We note that the exhaustive algorithm we developed earlier already works in this fashion.

An important assumption made here is that all possible bitvectors are used and pushed down as far as possible. An interesting direction for future work is to develop
optimization techniques that can find the optimal set of bitvectors to use for a given query using our framework. We believe our formalism can be easily extended to handle this. However, Theorem \ref{thm:optimality} does not hold in that case.

\subsection{Costing and Optimization: SJ and COM+SJ}
The formalism also naturally extends to the case when semijoins are used for full reduction, before doing the joins (i.e., when using the Yannakakis algorithm). To make
the discussion concrete, we consider a two-phase implementation, where the second phase produces the join results. Specifically, let $r$ denote the root of the tree. In
the first phase, we recursively reduce the relations using their children, starting with the parents of the leaves and ending with fully reducing $r$. For a relation
$p$ with children $c_1, \ldots, c_k$, for each of its tuples, we check whether it has a match in each of the children, and if not, the tuple is discarded. We assume
that these ``semi-joins'' are {\em exact}, i.e., there are no false positives. If approximate techniques (e.g., a bloomfilter) is used instead, the analysis is similar, but the cost formulas become more complex since we need to account for the false positive rate (as above).

At the end of the first phase, the root relation ($r$) is fully reduced, the leaves are not reduced at all, and the other relations are partially reduced. In the second phase, we directly produce the join results by starting with the fully reduced root relation, i.e., by using a left-deep plan with the root as the driver.  Note that, we do not explicitly reduce the other relations in this implementation; although some descriptions of the Yannakakis algorithm separate the second semi-join pass and result generation, that leads to unnecessary overhead with no change in guarantees.

Since we need to compute {\em adjusted} match probabilities and fanouts after reduction, we first derive a formula for that. Let $p, c$ denote a parent-child pair, with $m_{p \rightarrow c}$ and $fo_{p \rightarrow c}$ denoting the match probability and fanout when probing from $p$ into $c$. 
If $c$ is reduced by a factor of $ratio$ independently of $p$ (e.g., because of other semi-joins with its own children), then:
\begin{theorem}
\label{thm:adjusted_stats}
The adjusted match probability and fanout when probing from $p$ into reduced $c$ are:
$$m'_{p \rightarrow c} = m_{p \rightarrow c} \times (1 - (1 - ratio)^{fo_{p \rightarrow c}})$$
and
$$fo'_{p \rightarrow c} = fo_{p \rightarrow c} \times ratio / (1 - (1 - ratio)^{fo_{p \rightarrow c}})$$
\end{theorem}

These formulas hold as long as $c$ is reduced independently of $p$ (they rely on standard random sampling arguments). We also note that:\\[2pt] $s'_{p \rightarrow c} = m'_{p \rightarrow c} fo'_{p \rightarrow c}  = ratio \times m_{p \rightarrow c} fo_{p \rightarrow c} = ratio \times s_{p \rightarrow c}$, which matches the standard formulas for selectivity estimation.

Given this, we can estimate the costs of the two phases for the full reduction plan (Figure \ref{fig:full_reduction}) as follows.
For the first phase, the number of semi-join probes (assuming $R_2$ tuples are first semi-joined with $R_3$, and $R_1$ tuples are first semi-joined with $R_2$) can be estimated as:
$$|R_2| + m_3|R_2| + |R_5| + |R_1| + (1 - (1-m_3m_4)^{fo_2})m_2|R_1|$$
The first two terms capture the semi-joins of $R_2$ with $R_3$ and $R_4$, respectively. The third and fourth terms capture $R_5 \ltimes R_6$ and $R_1 \ltimes R_2$,
    respectively. The last term that captures $R_1 \ltimes R_5$ uses the above theorem ($ratio = m_3m_4$) to calculate the number of 
tuples that survive the semi-join with $R_2$.

In the second phase, the total number of probes depends on whether we are using the factorized representation (COM) or not (STD). We make three observations: (1) the match probabilities from the parent to child after the first step are all 1 (this is guaranteed due to the reduction step);
(2) the fanouts from parent to child can be computed as above using the reduced sizes of the children; and (3) $s_{p \rightarrow c} = m_{p \rightarrow c} fo_{p \rightarrow c} = fo_{p \rightarrow c}$. For example, the fanout and selectivity from $R_1$ to $R_2$ is:
$$fo'_{R_1 \rightarrow R_2} = fo_{R_1 \rightarrow R_2} \times ratio / (1 - (1 - ratio)^{fo_{p \rightarrow c}})$$
So the standard cost formula and the modified cost formula from Section \ref{sec:costing} can be used to estimate the number of probes (the formulas are much simpler since all the $m$ terms evaluate to 1).
We make the observation for COM that: 
\begin{theorem}
\label{thm:full-reduction}
If repeated probes are avoided, then the cost of the third phase is independent of the join order used in that phase.
\end{theorem}

\noindent\textbf{Optimization:} Finally, unlike the other techniques, finding the best full-reduction plan is much simpler. 
The three optimization decisions are: (1) which relation to use as the driver, (2) the order of semi-joins for each internal node (and root) in the first phase, and (3) the order of joins in the second phase.
For (1), we can try out all the relations as the driver and pick the best one. For (2), given a parent $p$ with children $c_1, \ldots, c_k$, they should be probed in the order of increasing $m'_{p \rightarrow c_i}$. 

The decisions for (3) depend on whether we are using COM or STD. For STD, the rank ordering algorithm will generate the optimal order, which in this case, is the order
of increasing $fo'$ (adjusted fanouts, since all selectivities are 1).
For COM, the optimal join order is to sort the relations by the product of the fanouts from the root to the relation. We defer the proof to the full version of the paper.

Note that, for any of the COM approaches, there is an extra ``expansion'' step at the end if we desire the output in the flat representation. We explicitly account 
for this cost when comparing a COM approach against an STD approach.

\subsection{Robustness Analysis: COM}
\label{sec:robust}
Next, we analyze the robustness of our cost model using an approach developed in~\cite{lip} for COM; similar analyses can be conducted for the other variants. An evaluation strategy $\mathcal{E}$ is said to be $\theta$-fragile and $\Theta$-robust with respect to a plan space $P$ if the maximum deviation in performance of any plan in $P$ from the best plan $\mathcal{E}_b$ is bounded between $\theta$ and $\Theta$.
The deviation is computed by normalizing the costs by the driver table cardinality (i.e., by considering the cost per tuple) and by the spread of the join operator selectivities ($s_{max} - s_{min}$).  
It is shown in~\cite{lip} that, for the simpler evaluation approach, for a star query with $n$ dimension tables, $$\theta = \dfrac{1-s_{min}^{n-1}}{1-s_{min}}$$ 
That paper doesn't show it, but we can also derive:
$$\Theta = \dfrac{1}{s_{max}-s_{min}}\sum^{n-2}_{i=1} s_{max}^i-s_{min}^i$$

Making a similar argument, we can show that our approach further reduces the fragility to plan selection by narrowing the spread between the best and worst plan making the lower bound $\theta$ smaller. Specifically, we can show that: $$\theta = \dfrac{1-m_{min}^{n-1}}{1-m_{min}}$$ where $m_{min}$ is the smallest match probability. 
Similarly, we show that (proofs can be found in~\cite{optimizing_m-n_joins}):
$$\Theta =  \dfrac{1}{m_{max}-m_{min}}\sum^{n-2}_{i=1} m_{max}^i-m_{min}^i
$$

To better understand the empirical implications of these results, we ran a series of simulation experiments using a star query with 10 relations, where we introduce variations in the selectivities from optimization to execution time. For each set of runs, $m_i$ and $fo_i$ were selected uniformly at random from a set of ranges ($[0.05-0.2], [0.05-0.5], [0.1-0.5]$, and $[0.5-0.9]$ for $m_i$, and $[1-2], [1-10]$, and $[10-100]$ for $fo_i$). $s_i$ is always set to be $m_i \times fo_i$. We experimented with introducing a low degree of variation (in the range of 15-20\%) and a high degree of variation (90-95\%). Thus the latter models a situation where the selectivity estimates are really off, stressing the robustness of the approach.

\begin{figure}
\centering
    \includegraphics[width=.45\textwidth]{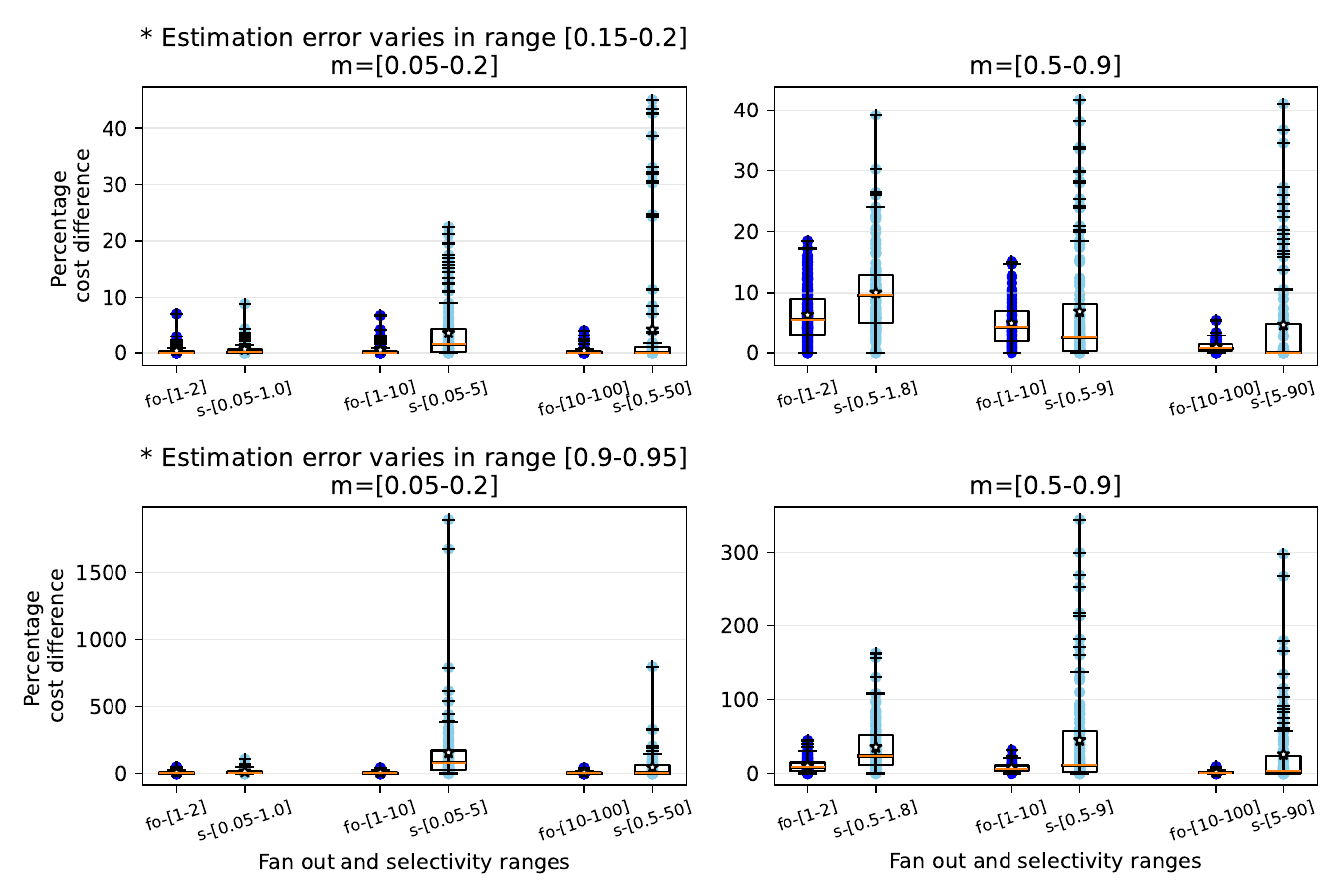}
    \caption{Variation of the percentage of cost difference between actual best plan and the estimated best plan for two different estimation error ranges}
    \label{fig:boxplot_all_er}
\end{figure}

For each set of ranges, we generated 100 samples uniformly at random, and measured the cost difference between the best join order based on the actual selectivities and the best join order based on the estimated costs for both cost models. Figure~\ref{fig:boxplot_all_er} shows these results, with the top 2 plots depicting the low variation scenario. It is evident from comparing the top plots with the bottom plots that selectivity-based cost models are impacted greatly by escalating estimation errors. Across all the match probability ranges and all the fan-out ranges, the new cost model is more robust against estimation errors compared to the cost model based on selectivities. Also, we notice that when the fanout is in the range $[1-2]$, both models show similar behavior to estimation errors since $s_i$ is off only by at most a factor of 2 compared to $m_i$ in that case.

To summarize, this robustness analysis and the empirical simulation illustrate that, in presence of many-to-many joins, it is critical to properly account for the repeated probes, both in the implementation of the join algorithms, and in cost modeling for optimization. Without this, the differences between different plans may be highly exaggerated, especially in presence of high fanout (exploding) joins, in turn exaggerating the benefits of more complex query optimization or selectivity estimation techniques.

\section{Vectorized Query Execution Architecture} \label{vectorized execution}
In this section, we sketch the query execution architecture of our prototype that we have built for comparing these approaches, specifically, a modified hash join operator as well as a vectorized intermediate result representation, that enables us to avoid redundant probes 
during a one-pass execution. We base our hash join implementation on a vectorized hash join implementation used in DuckDB~\cite{duckdb2019}. 
We then present extensions to incorporate bitvector-based early pruning as well as semijoin-based full reduction in our implementation.

\subsection{Overview}
\label{sec:overview architecture}
DuckDB uses a vectorized interpreted execution engine that executes the query in a ``Vector-Volcano'' model, using vectors of a fixed maximum size (by default 1024). While fixed-length types like integers are stored as native arrays, variable length types like strings are represented as a native array of pointers pointing into a separate string heap. Null values are represented using a separate bitvector. Data
vectors may also have a selection vector that indicates the offsets into the vector stating the relevant indices in the vector (e.g., in a selection or to indicate values that participate in a join/probe). DuckDB supports a vast number of relational vector operations including relational joins which operate on vectors of data instead of doing tuple-at-a-time processing~\cite{duckdb2019}.

Our prototype uses a similar vectorized interpreted execution engine, but allows switching  between {\bf six different approaches}. First, it allows using either: (1) the standard execution model (STD), that fully constructs all 
the intermediate tuples after each join (as a vector of vectors), or (2) the factorized execution model (COM) that uses a factorized intermediate representation to which additional vectors are appended after each join. COM needs
a final ``expansion'' step to generate all the result tuples at the end. For either COM or STD, we can turn on either bitvector-based early pruning (Section~\ref{sec:bitvector}) or full semijoin-based reduction
(Section~\ref{sec:fullreduction}), or leave both off (giving us a total of 6 options).  
Each relation is loaded into memory as {\em DataChunks} consisting of 2048 tuples by default, with each DataChunk consisting of as many vectors as the number of
attributes in the relation.
The joins are executed in a pipelined fashion batch-at-a-time, where the results of one join (for a batch) are immediately emitted to the subsequent joins. 
The execution engine supports vectorized hash join operations using SIMD instructions like non-contiguous loads (gathers) and stores (scatters). Vector processors take advantage of data-level parallelism by using multiple replicated lanes in the vector unit. Each lane, called a vector lane, includes a portion of the vector register file, a data path from the vector functional unit, and one or more ports to access
memory~\cite{vector_threading}. We implemented a general approach called {\em vertical vectorization}, where each vector channel processes a different input key during
the join. Each vector lane handles a key and accesses a specific position in the hash table~\cite{SIMD}. We expand on these ideas in the following sections.

\begin{figure}%
\hspace*{-0.35in} \includegraphics[scale=0.55]{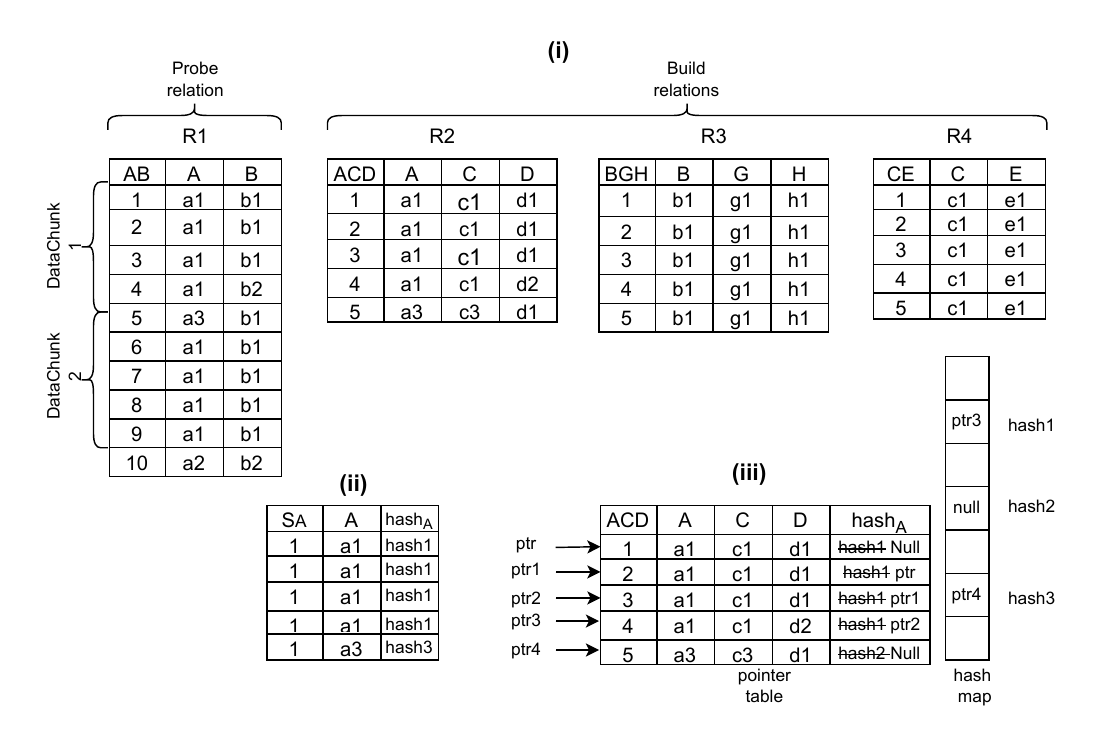}
  \caption{(i) Simpler running example for query execution with query plan: $R_1\bowtie R_2\bowtie R_4\bowtie R_3$; (ii) Probing structure for \nth{1} chunk of $R_1(AB,A,B)$; (iii) Build table for $R_2$ relation.}
    \label{fig:running example}
\end{figure}

\subsection{Data Structures}
\noindent\textbf{\textit{Running Example:}}
Figure \ref{fig:running example} shows a simpler example with 4 relations that we use to illustrate the data structures and the algorithm. We have 4 relations $R_1(AB,A,B)$ (driver), $R_2(ACD,A,C,D)$,\\ $R_3(BGH,B,G,H)$, and $R_4(CE,C,E)$, where we use $AB$, $ACD$, etc., to denote the ID columns. 
For illustration, we use the join order: $R_1(AB,A,B) \bowtie R_2(ACD,A,C,D) \bowtie R_4(CE,C,E) \bowtie R_3(BGH,B,G,H)$
Figure \ref{fig:Factorized representation} shows a COM chunk after the first join.

\sskip
\noindent
The main data structures in our engine are as follows. \\
\noindent\textbf{DataChunk:} A DataChunk contains a batch of tuples, either from a base relation or from an intermediate result, and corresponds to a vector of {\bf VectorColumns}. Each base relation is partitioned into multiple DataChunks. The base relation DataChunks as well as the intermediate DataChunks for STD contain VectorColumns of equal size, whereas the intermediate DataChunks for COM typically have
VectorColumns of different sizes. We assume the existence of an explicit ID column with unique values -- these are not necessary for the implementation, but make the exposition easier. Assuming an initial DataChunk size of 5 tuples, 2 DataChunks (as shown) are created for $R_1$ in our example. \\
\textbf{VectorColumn:} A VectorColumn typically corresponds to a relation column (i.e., attribute). In COM, some of the VectorColumns that correspond to ID attributes may also have selection vectors with one-to-one correspondence with the corresponding rows ($S_{AB}$ in Figure \ref{fig:Factorized representation}). The selection vector determines whether the corresponding values participate in a join operation.\\ 
\textbf{Count VectorColumn}: A count VectorColumn is added to the COM DataChunks after a join. For each value in the join column of the probe relation, there is a corresponding entry in the count VectorColumn (e.g., $C_{ACD}$), indicating how many matches were found in the probe relation (i.e., fanout for that value). 
For an entry in the Count VectorColumn, the corresponding value in another vector, called Prefix\_Sum\_Count (not shown in the figure), maintains the sum of all the counts of all the entries above it. This vector is essential for the result generation from the factorized representation. 

\sskip
\noindent\textbf{Hash Table:} A Hash Table consists of a pointer table and a hash map; an example is shown in Figure \ref{fig:running example}(iii) for $R_2$, with $A$ being the search key (probe attribute). 
The pointer table contains the tuples from the build relation ($R_2$), with an additional column used to create a chaining structure ($hash_A$ in the example). 
We can see such a chain for $A = a_1$ in the figure.
Say $a_1$ is mapped to location $hash1$ in the hashmap. That location stores a pointer to the {\em last} tuple seen during the build for that hash value, i.e., $ptr3$. The $hash_A$ column for $ptr3$ entry in the pointer table enables us to find $ptr2$, another tuple that also hashes to the same location, and subsequently $ptr1$ and $ptr$. $hash_A$ value for $ptr$ is set to Null, indicating the last entry in the chain.

Analogously a probing structure is built during the {\em probe} step, a DataChunk at a time (for $R_1$ as shown in Figure \ref{fig:running example}(ii)). The probing structure contains the join key VectorColumn(s) and the corresponding hash values.

\subsection{Hash Join Execution}
\noindent
Next, we discuss how a hash join is executed.

\sskip
\noindent\textbf{Building Hash Tables:} We process multiple input keys/tuples from the build relation simultaneously using vector lanes. 
After computing the hash values for the input keys, we scatter the keys and the corresponding pointers to store them in the hash map, that keeps track of the input keys and the associated chain of matching tuples.
Updating the hashmap and creating the chain-like structure that links all matching tuples with the same hash value is a sequential process, done in a tight for-loop to ensure efficiency and maintain the integrity of the hashmap.

\begin{figure}
    \centering
    \includegraphics[width=.45\textwidth,viewport=0 0 409 379.92,clip]{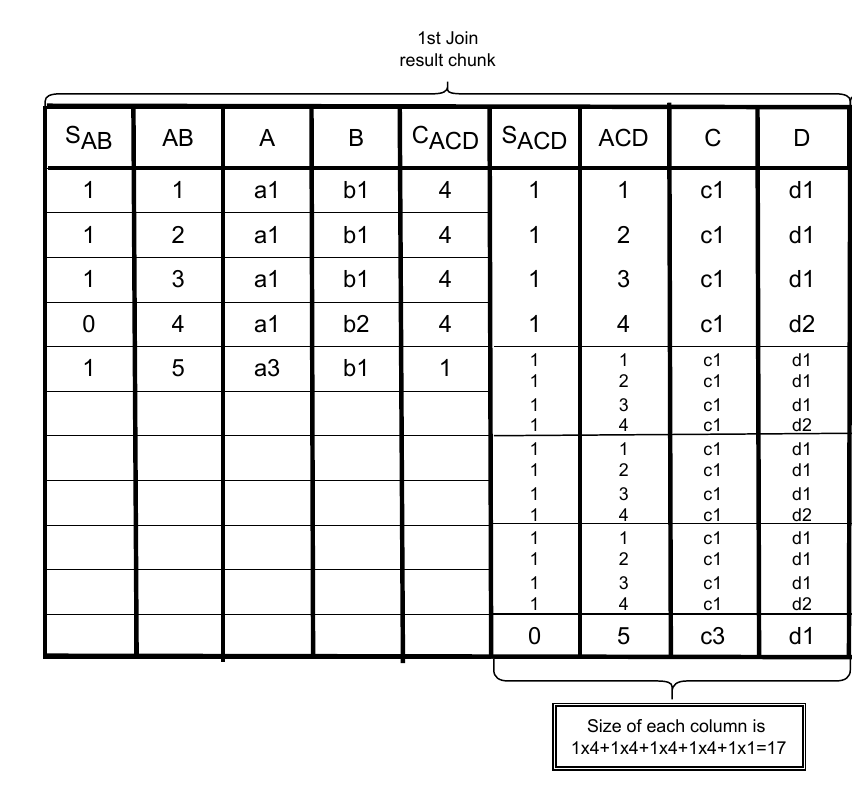}
    \caption{Partial factorized representation for query in Fig. \ref{fig:running example}}
    \label{fig:Factorized representation}
\end{figure}

\sskip
\noindent\textbf{Vectorized Linear Probing: }
Linear probing traverses the table linearly (i.e., follows the chains) until it reaches a $null$. The input keys in the probe relation DataChunk are 
processed using all the available SIMD vector lanes iteratively, and within each iteration, we 
use a nested loop to find all matches per input key. The longest chain 
becomes the bottleneck here, potentially causing under-utilization of vector lanes.
However, this enables us to capture the number of matches per input key (for Count VectorColumn) as we show next.

\sskip
\noindent\textbf{Hash Join Operator:}
Our hash join operator takes a DataChunk as an input and emits a DataChunk as the output. We probe into the hashmap of the build table and follow the pointers until we obtain all
the matching tuples for those values in the probing structure where selection vector is 1. For instance, the value of the selection vector is non-zero for the first tuple
of $R_1$'s first DataChunk. We start by following $ptr3$ which is stored in the $hash1$ location in $R_2$'s hash map. We determine whether the values in the columns match, and if they do, we increment the number of matches for the value in $R_1$. The count VectorColumn holds the value of this count. Additionally, the result chunk is augmented with the values from the payload columns, i.e., the remaining columns from the build table besides the join key column. For each match found for the join column entry, we add the corresponding entries for the payload columns. Accessing the $hash_A$ column of the tuple pointed by $ptr3$ allows us to go on to the next matching tuple pointer. If it is null, the probe has been completed. Now that we have $ptr2$ in $hash_A$ column of the tuple pointed by $ptr3$, we proceed in the same manner by checking the tuple referred by $ptr2$ to see if it matches the tuple of $R1$.  
This continues until a tuple with the $hash_A$ column value Null 
is encountered. The number of matches found is kept in the count column. As shown in Figure \ref{fig:Factorized representation}, the resulting chunk that is formed after the first join contains VectorColumns for each column contributed by both the probe and build relations, as well as the matching count column that has a one-to-one mapping with the entries from the probe relations.

The partial factorized result in Figure \ref{fig:Factorized representation} can be interpreted as follows. 
The last tuple in $R_1$ DataChunk found one match in $R_2$ ($C_{ACD} = 1$), while the first four tuples from that DataChunk found four matches each in $R_2$.
Thus the first four entries in columns $C$ and $D$ correspond to the first tuple from $R_1$ DataChunk, the next four entries correspond to the second tuple, and so on.
This result DataChunk is then probed against $R_4$'s build table. VectorColumns $Count CE$, $CE$, and $E$ are added to the resulting chunk after the join operation (not shown
  in the figure). A one-to-one correspondence exists between the join key $C$ VectorColumn and the $Count CE$ VectorColumn. We set the selection vector value to 0 for C values that failed to find matches in the join so that we can ignore them in subsequent joins. The updated resulting chunk is then probed against relation $R_3$. Following the join, the resultant chunk is expanded to include VectorColumns $Count BGH$, $BGH$, $G$, and $H$. 

\sskip
\noindent\textbf{Result Expansion:} After the final join has been performed, the result DataChunk may need to be expanded to produce the flattened tuples. We currently do this in a {\em depth-first} manner in our prototype since that is more memory-efficient.  
We begin with the driver relation columns in the chunk ($A$ and $B$ in the example). For each row (where those columns are present), we horizontally expand the factorized representation by maintaining a {\em row index vector}, that keeps track of the status of the expansion, and a {\em row result vector} which maintains the current partially expanded tuple. Considering the example in Figure \ref{fig:Factorized representation}, we begin with $(1, a_1, b_1)$. Since $SC_{ACD} = 1$ and $C_{ACD} = 4$, this tuple joins with 4 tuples from $R_2$, and will be expanded to $(1, a_1, b_1, 1, c_1, d_1)$, $(1, a_1, b_1, 2, c_1, d_1)$, $(1, a_1, b_1, 3, c_1, d_1)$, $(1, a_1, b_1, 4, c_1, d_2)$ one at a time. Next, the first of those tuples $(1, a_1, b_1, 1, c_1, d_1)$ will be expanded to 5 tuples (e.g., $(1, a_1, b_1, 1, c_1, d_1, 1, e_1)$). Finally, $(1, a_1, b_1, 1, c_1, d_1, 1, e_1)$ in turn will be expanded to 5 tuples, the first being  $(1, a_1, b_1, 1, c_1, d_1, 1, e_1, 1, g_1, h_1)$. Once the last column is added (or we encounter a selection vector entry of 0), we backtrack to the previous step and expand along the next tuple. The row index vector keeps track of how many of these expansions have been done at each step. Figure \ref{fig:Result expansion} illustrates this process.

\begin{figure}[h]
    \includegraphics[scale=0.49]{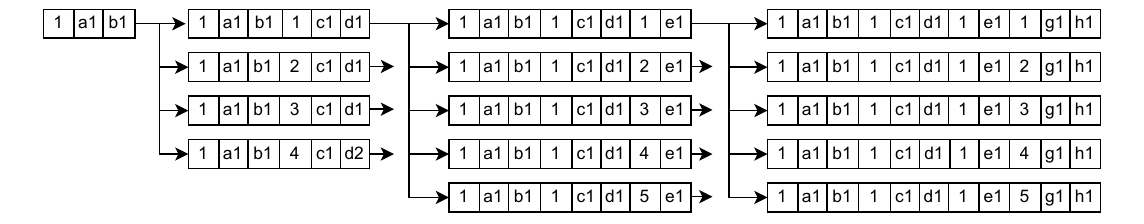}
    \caption{Final result expansion can be done in depth-first or breadth-first fashion}
    \label{fig:Result expansion}
\end{figure}

The result expansion can also be done in a {\em breadth-first} manner which is somewhat more parallelizable, but requires a separate sequential ``counting'' step to count the number of result tuples for each input tuple so that an appropriate number of copies can be made for each tuple. We are planning to implement and experiment with that approach in future work.

\sskip
\noindent\textbf{COM vs STD Implementations:}
Both COM and STD implementations are vectorized, and the build phase is identical. However, during the join phase, STD fully materializes the intermediate results
whereas COM maintains the intermediate DataChunks in a factorized representation. In either case, {\em selection vectors} are used to keep track of which tuples
contribute to the final result (and are set to False by any of the probe operations if there are no matches). Since the final DataChunk contains full result tuples, the
expansion phase is not needed for STD, but is needed for COM adding extra overhead at the end.

\subsection{Bitvector-based Early Pruning}
\label{sec:bitvector}
We integrated bitvector-based early pruning, following the approach in~\cite{bitvector_aware_query_proc}, 
to evaluate how well it achieves robustness alone. 
Each hash join operator creates a single
bitvector filter 
from the equi-join column on the build side as the keys of the bitvector filter.  
A bitvector filter is pushed down to the lowest possible level on the subtree rooted at the probe side to eliminate tuples from that subtree as quickly as possible. 
At each join operator, the matches are pruned by probing against all the corresponding bitvectors in a vectorized fashion. However, we might get false positives here because we only do a hash comparison of the values. In left-deep plan space, 
for the first operator only, bitvectors will be pushed down to both the probe side and the build side (so that the driver tuples can be pruned using the appropriate bitvectors before any probes). From the second operator onwards
in the plan tree, bitvectors are only pushed down to the build side since the probe side has already been pruned by all bitvectors that can be applied to it by a previous join operator.

\subsection{Semijoin-based Full Reduction}
\label{sec:fullreduction}
To compare with previous related work (e.g.,~\cite{Tziavelis2019OptimalAF,TziavelisGR21}) that achieves full reduction to remove all dangling tuples in the driver relation that don't contribute to final output, we integrated vectorized full
reduction into both STD and COM implementations. There is no extra hash table build cost since only parents need to probe into children in the join tree in bottom-up fashion to fully reduce the driver relation. Once all the dangling tuples of the driver relation are removed, the final join processing and result generation are done identically to STD or COM. 

\subsection{Comparisons to Prior Work}
We briefly compare our representation with two other prior implementations, namely FDB~\cite{FDB}, and GraphFlowDB~\cite{graphflowdb}.

FDB was the first work to systematically explore the use of compressed factorized representations for processing multi-way join queries, including cyclic queries~\cite{FDB,olt,danolt}. That work proposes two factorized representations, captured by the notion of {\em f-trees} and {\em d-trees}. Our representation can be seen as equivalent to an {\em f-representation}, corresponding to a {\em f-tree} rooted at the driver relation. One key difference is that, our representation works at the level of tuples rather than attributes. In other words, each node in an f-tree corresponds to an attribute; in our representation, each node (in essence) corresponds to a subset of attributes from a relation. 
More concretely, the representation corresponding to the first driver tuple in Figure \ref{fig:Factorized representation} can be written using the terminology from that work as:
\begin{align*}
\langle 1, a1, b1 \rangle &\times ((\langle 1, c1, d1 \rangle \times (\langle 1, e_1 \rangle \cup \langle 2, e_1 \rangle ...) \mbox{\ \hspace{2in}} \\
       &\mbox{\ \hspace{0.5in}}\cup  \langle 2, c1, d1 \rangle \times (\langle 1, e_1 \rangle \cup \langle 2, e_1 \rangle ...) \cup \cdots) \mbox{\ \hspace{2in}} \\
       &\times (\langle 1, g_1, h_1 \rangle \cup \langle 2, g_1, h_1 \rangle \cdots)
\end{align*}

The focus of that line of work, however, was primarily on the representation and compression issues, and they did not consider optimization and execution issues in as much depth, and do not
provide a mechanism to model the execution cost when using factorized representations. Further, their implementation uses a hierarchical data structure to store the base as well as intermediate relations, 
where each level corresponds to an attribute; the children of each node are grouped and maintained in a sorted order. Such a data structure is not easy to incorporate into a traditional or columnar-based storage or query engine. Our approach offers a more efficient and practical way to get those same benefits, especially given our use of vectorization. 

GraphflowDB~\cite{graphflowdb} is an in-memory property graph DBMS that does list-based execution while adopting a factorized structure for the intermediate relations. 
Although similar at a high level, our implementation differs from their in several important details.
First, their approach works off of an {\em adjacency list} representation of the data (equivalent to {\em join indexes}) where for each node (tuple), the IDs of its neighbors (i.e., the tuples it joins with) are maintained as a list. This does not
hold true for most relational storage engines, and would require an expensive pre-processing step (that in effect does all the pairwise joins). Second, although the data is passed between operators in chunks (vectors), the join itself sequentially loops through the tuples and propagates each tuple through the entire rest of the pipeline before moving on to the next tuple in the chunk, whereas our join operator is a standard
vectorized chunk based hash join operator that does morsel-driven execution, modified to work on the factorized representation.

\begin{figure}
\centering
    \includegraphics[scale=0.35]{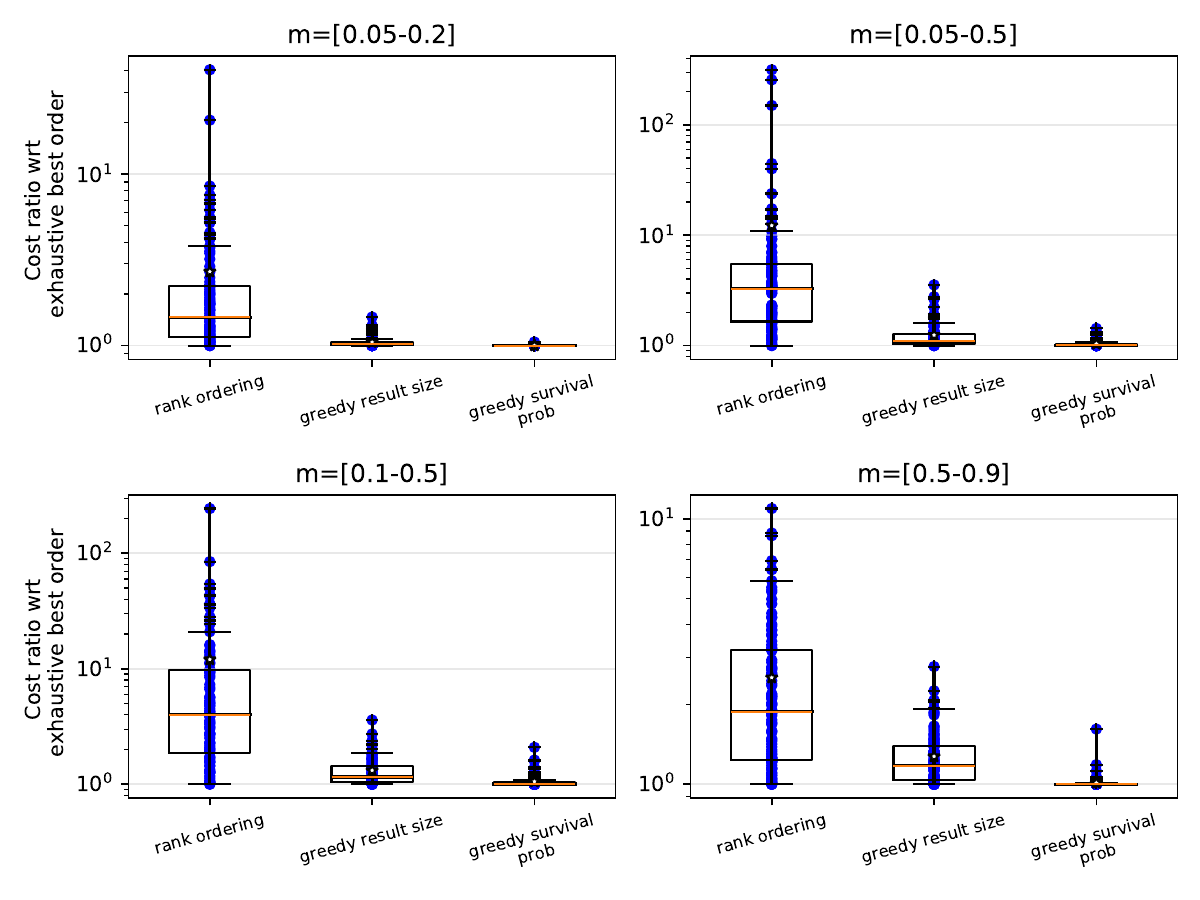}
    \caption{Comparing the join order optimization algorithms} 
    \label{fig:boxplot_all_opt}
\end{figure}

\section{Experimental Evaluation} \label{evaluation}
In this section, we present a comprehensive experimental evaluation that validates the claims made in this paper. 
We compare the six approaches outlined in Section~\ref{sec:overview architecture} to isolate the distinctive differences of each implementation and its impact to the robustness of the join ordering. 
The key results of our
evaluation are: (a) in presence of many-to-many joins, the standard rank ordering heuristic (which is very similar to what modern query optimizers do) 
finds orders-of-magnitude worse plans compared to our proposed greedy heuristics; (b) avoiding redundant probes through use of our approach sometimes results in orders of magnitude performance improvements, both in wall-clock time as well as in the number of probes; 
(c) bitvector-based early pruning or two-pass semijoin full reduction, by themselves, do not perform as well as COM, but often provide the overall best performance when combined with COM, 
(d) our cost model tracks real execution costs and can be used for making optimization decisions among the competing approaches. We also performed
an extensive evaluation of the robustness of these techniques to the choice of the join order, using both simulations and our prototype,  but omit those due to space constraints.

We implemented our prototype of the vectorized query execution engine in C++ 
closely following the standard practices. 
For all experiments, the initial chunk size was set to 2048.
Main evaluation metrics we focus on are: (a) the CPU wall clock time, and (b) the number of probes into the hash tables. As discussed earlier, latter may be a more important metric to focus on in many situations. We use a machine with 20 CPU cores and 32GB RAM for all experiments. 
\begin{figure}
\centering
    \includegraphics[width=0.5\textwidth]{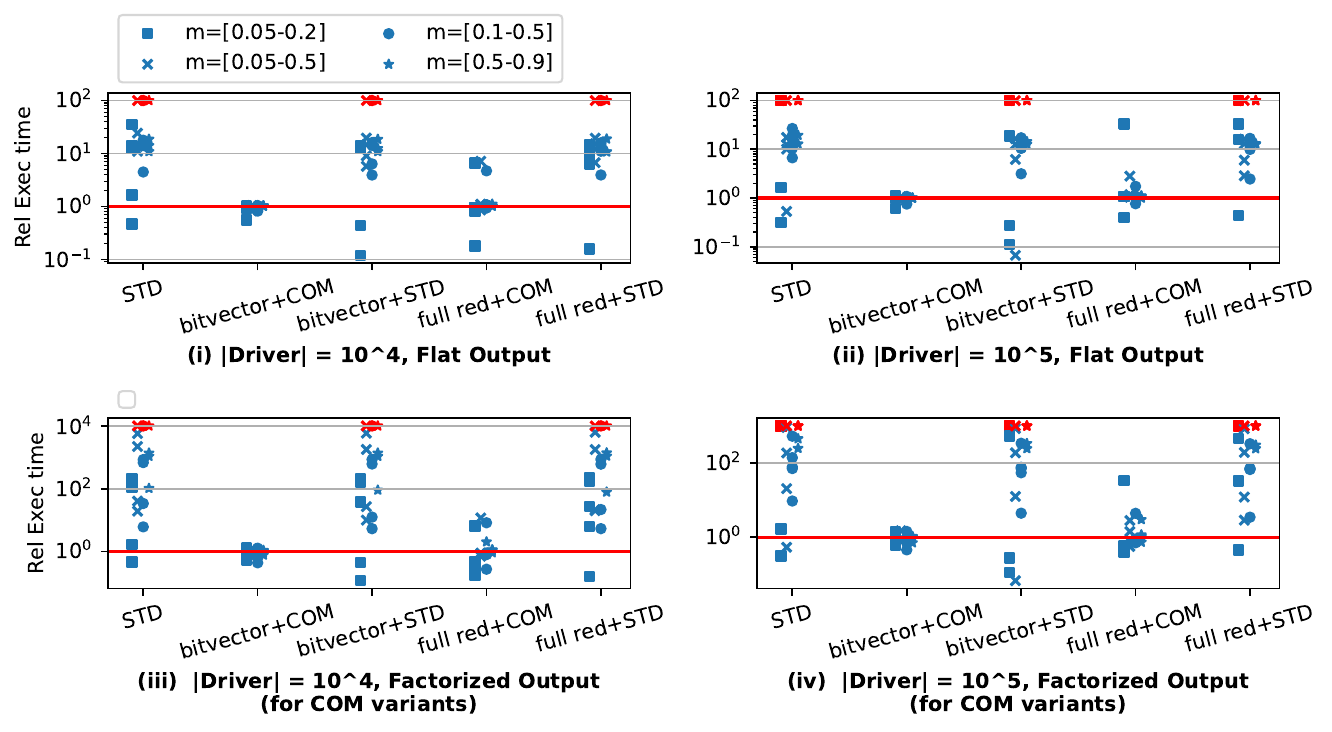}
    \caption{Relative Runtime comparison for 5 different approaches w.r.t. COM for synthetic data set} 
    \label{fig:synthetic-2}
\end{figure}

\subsection{Join Order Optimization}
We begin with comparing the three greedy heuristics discussed in Section~\ref{costmodel} for join order optimization with the exhaustive algorithm. We randomly generate 
join trees with up to 20 nodes, and compare the 
estimated execution costs of the plans found by the four algorithms. The fanouts ($fo$s) were randomly chosen between $[1-10]$ and match probabilities were selected uniformly at random from four sets of ranges: $[0.05-0.2], [0.05-0.5], [0.1-0.5]$, and $[0.5-0.9]$. We generate 100 sample join trees for each match probability range. For each sample join tree, the number of children for the root node could vary from $[2-5]$; for the other nodes, it is $[0-3]$. We ran a number of other experiments, with different settings, and the results were consistent. Figure \ref{fig:boxplot_all_opt} shows the results as a box plot. We show the ratio cost of the best order determined by each heuristic with respect to the cost of the best join order determined by the exhaustive algorithm. 
We see that across all match probability ranges, the survival probability-based heuristic is closest to the optimal algorithm compared to the other two approaches. The rank ordering heuristic in particular performs the worst, sometimes by orders of magnitude. We note that the execution cost is determined in all instances under the assumption that redundant probes are avoided (i.e., using the cost model we developed). This clearly illustrates the need to properly account for those during query optimization; the rank ordering algorithm that uses a single number ($s_i$) to model a join operator does not suffice even as a heuristic.

\subsection{Synthetic Benchmark} 
To properly evaluate the different approaches under a varying set of parameters, we create a synthetic benchmark modeled after the Star Schema Benchmark (SSBM)~\cite{ssbm} that contains a  number of synthetically
generated datasets and a mixture of different types of queries, including star queries, path queries, and snowflake queries (SSBM only contains star queries, which form a trivial special case for these approaches, especially from the
query optimization perspective).
We show experimental results for 4 query shapes: (a) a {\em 7-relation star} query, (b) a {\em 11-relation path query} (with the center relation as the driver), (c) a {\em 3-2 snowflake query}, where the driver (center) relation has 3 children each of which has 2 children in turn,
and (d) a {\em 5-1 snowflake query}, where the driver relation has 5 children, each of which has 1 child. These query shapes capture the spectrum of query shapes that we typically see in practice.

For each dataset, the match probabilities were chosen uniformly at random from one of the following ranges: $[0.05-0.2], [0.05-0.5], [0.1-0.5]$, and $[0.5-0.9]$, and the fanouts 
were chosen uniformly at random between $[1-10]$. We experiment with 3 different driver relation sizes: $10^4$, $10^5$, and $10^6$,
and with 16 queries covering different match probability ranges and query types.  

Figure \ref{fig:synthetic-2} shows the results of the 5 approaches (excluding COM) normalized using the execution cost of COM (shown as the red horizontal baseline). The queries which timed out are shown as red data points at the top of the plot (all for STD and its variants). 
The top 2 plots show the relative execution times when the output is in flattened format for all approaches, whereas the bottom plots show the scenario when the output is in factorized format for COM variants. 
The best survival probability-based ordering is chosen as the default order for all queries. 

The results conclusively demonstrate the need for avoiding redundant probes through use of factorized representation, even if the final output is in flattened form. We see that for almost all queries, COM variants outperform STD variants, by orders of magnitude in many cases (we note again that a number of queries timed out for STD).
We also see that bitvector or full reduction, by themselves, are not competitive with COM, highlighting the need to use COM in addition to those approaches. 
In a small fraction of cases, we see STD performing better than COM. This usually happens when the match probabilities are very low for several of the joins, resulting in a huge reduction in the number of tuples quickly. In these cases, the overheads of COM dominate the overall execution time.
Our cost model, through use of match probabilities and fanouts, can be used to identify such cases easily.

Finally, among the three COM variants, we see that COM and COM+BVP perform very similarly, with the latter showing slightly superior performance. Full reduction
using semijoins has a higher variance, and sometimes performs orders-of-magnitude worse, especially for path or snowflake queries. The primary reason for this is
usually the presence of a highly selective join condition. In standard left-deep plans, the optimal plan may perform that join first to significantly reduce the 
intermediate result size early on. However, the benefits of such a highly-selective join are lower in case of SJ because many of the semi-joins are done independently
of each other. Once again, our optimization algorithms can be used to identify such situations and choose the right technique for a given query.

\begin{figure}
    \centering
    \begin{subfigure}{\linewidth}
        \centering
        \includegraphics[scale=0.45]{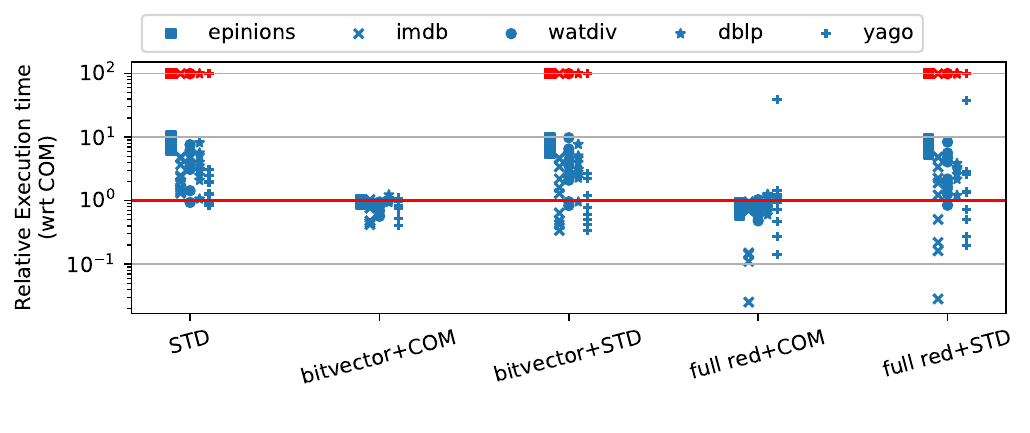}
        \caption{Flat output format}
        \label{fig:ce-benchmark}
    \end{subfigure}
    \begin{subfigure}{\linewidth}
        \centering
        \includegraphics[scale=0.45]{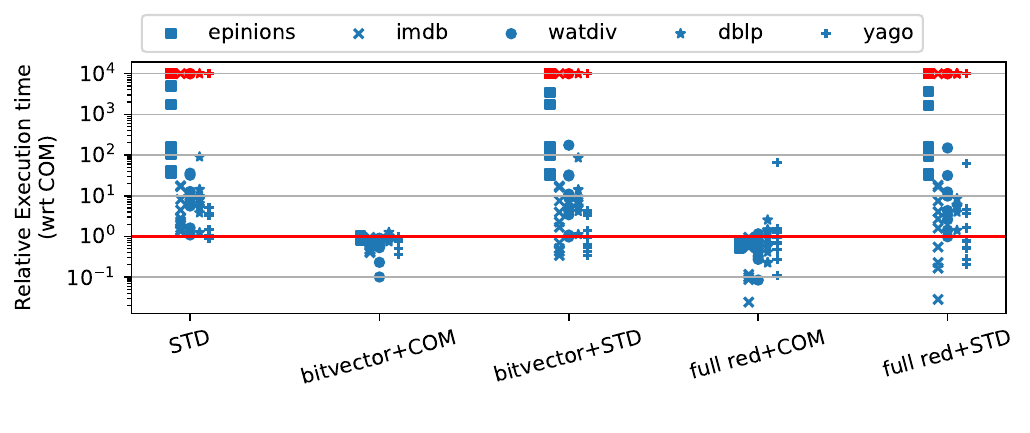}
        \caption{Factorized output format (for COM variants)}
        \label{fig:ce-benchmark-excluding-result-gen}
    \end{subfigure}
    \caption{Relative runtime comparison of the five approaches vs COM for CE benchmark}
    \label{fig:ce-combined}
\end{figure}

\subsection{CE Benchmark}
Next, we do a similar analysis using the CE benchmark~\cite{chen2022accurate, diamond}\footnote{\url{https://github.com/cetechreport/CEExperiments}}, which contains a wide range of complex queries that exhibit intermediate result size explosion due to many-to-many joins. 
We use 5 datasets from the CE benchmark: {\em epinions}, {\em imdb}, {\em watdiv}, {\em dblp}, and {\em yago}. We filtered the queries such that the result sizes don't exceed $10^{10}$. 
From each dataset, we selected 10 queries randomly. Figure \ref{fig:ce-benchmark} shows the relative execution times with respect to COM. As above, COM baseline is shown as the horizontal red line, and the data points which timed out are shown in red color. 

Overall, the results from these experiments are very similar to the ones above. For some queries in {\em imdb} and {\em yago}, COM overheads dominate because the fanouts are low. However, for almost all the queries, COM variants outperform STD variants, sometimes by orders of magnitude.
COM approaches are also more memory-efficient in general. Here we see COM, COM+BVP and COM+SJ performing very similarly, with the latter somewhat better (but also showing
higher variance). 

These results, as above, again highlight the need to accurately compute the costs of different plans that combine the different techniques in various ways, using a formal cost model such as ours, to pick the best plan for a given query and dataset.

\subsection{Simulation Analysis}
Our cost formulas can be used to analytically compare the different approaches under a variety of settings. We illustrate
one such analysis here, where we compared the performance of the techniques when all the relations are the same size and
all the match probabilities and fanouts are identical. This is an idealized setting that is unlikely to be seen in practice,
but helps in understanding the relative strengths and weaknesses of the techniques. 
Figure \ref{fig:best-cost-6-approaches} shows the results of this experiment. For the four query shapes (as before),
we plot the total running time of the 5 approaches as the match probability $m$ is varied from 0.1 to 0.9, for two
values of fanouts, 2 and 5 (we omit STD by itself because its high execution costs distort the plots). 

For BVP and SJ approaches, we model the cost of a semi-join probe or a bloomfilter-probe as 1/2 that of a normal probe (i.e., 
if a plan made 1000 bloomfilter probes, we count it as 500 probes). These ``weight parameters'' depend on the computational
environment and can be estimated (as we did) through micro-benchmarking; i.e., for a set of queries, we measured the costs
of the different types of probes to find the appropriate value of this parameter (this parameter depends on the sizes of the
hash tables and bloomfilters, but overall, the optimization algorithms are not highly sensitive to it). We similarly 
measured the cost of result generation, and model the cost of a single tuple generation as 1/14th of the cost of a hash join
probe. This may seem quite significant, but is a result of our highly efficient, vectorized, implementation of the hash join.
For any of the COM variants, the final result size is used to compute the cost of the expansion phase (by multiplying that
        by 1/14).

\begin{figure}
    \centering
    \includegraphics[width=1\linewidth]{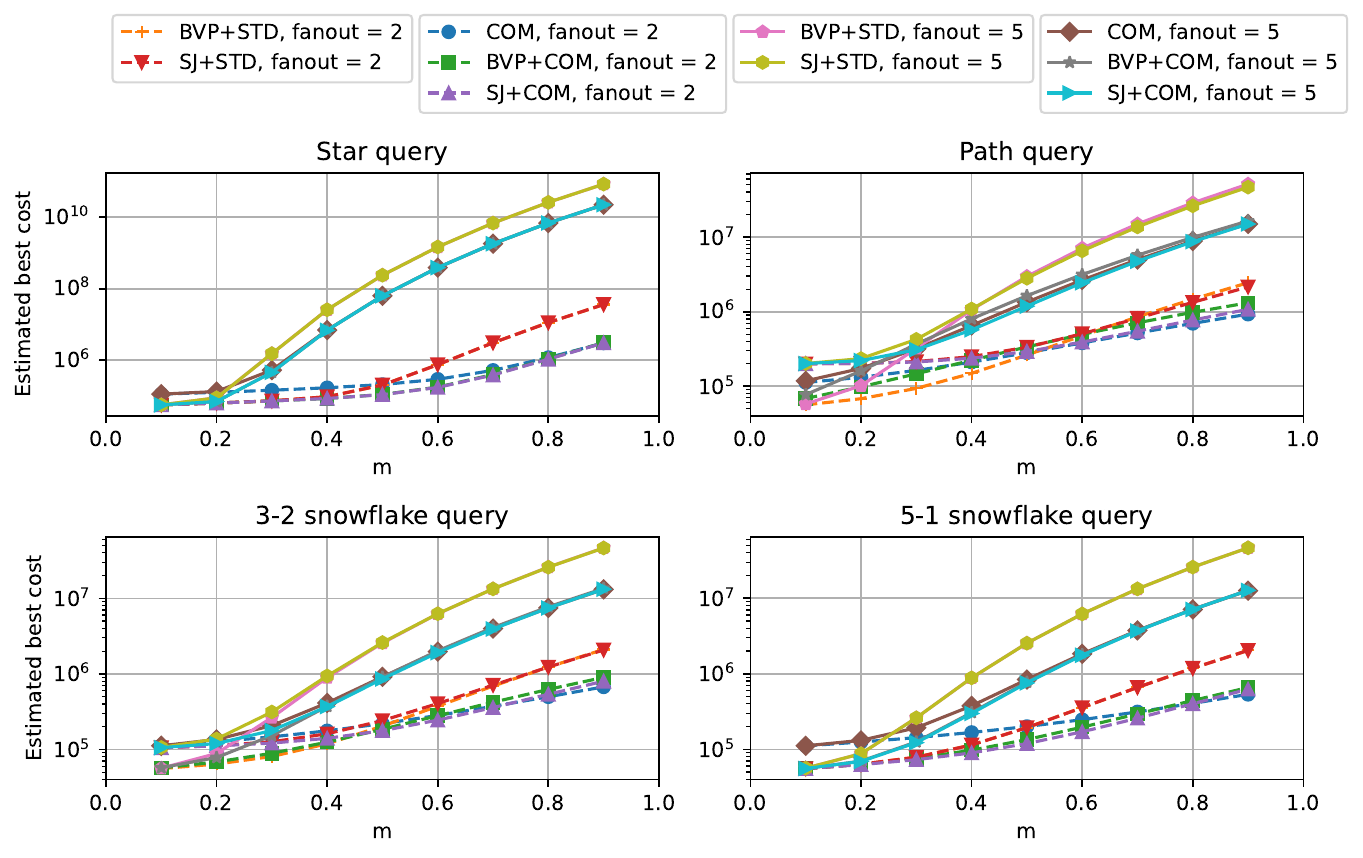}
    \caption{Simulation Analysis}
    \label{fig:best-cost-6-approaches}
\end{figure}

We see that irrespective of the fanout value, STD variants are competitive with COM when the match probabilities are low; however, the gap between them 
increases rapidly as the match probabilities increase, especially when the fanouts are also high. The relative performance of the COM variants shows several
interesting phenomena. At the low match probabilities, COM+BVP performs the best since it is able to eliminate the tuples fast using the bloomfilters. 
We see this especially with path queries where BVP+STD performs the best, but we see this with the snowflake queries as well to some extent.

As the match probabilities increase, the bloomfilters don't help much with eliminating tuples and the overheads of those start showing up. At the highest
match probabilities, we see that COM performs the best by itself, as it avoids the cost of bloomfilter or semijoin probes, which are useless in that setting.

Overall this analysis illustrates the benefits of our formal framework in better understanding the different algorithms under different settings.

\subsection{Cost Model Validation}
Finally, we address the question of how well our cost model tracks the actual execution costs that we see. We use five queries from the synthetic benchmark of different shapes, so that we can control 
the match probabilities and fanouts, and ensure the independence and uniformity properties. Figure \ref{fig:cost-model} shows a scatter plot of the predicted costs (in terms of 
the expected number of probes per driver tuple) vs the actual execution cost (in seconds), for 300 randomly chosen join orders for each query, for $10^5$ driver tuples. 
We use the same weight parameters as above (1/2 and 1/14) to find the cost of a plan in terms of the number of hash join probes.
As we can see, 
 the predicted costs align very well with the actual execution times across different query shapes, further validating the use of the cost model for query optimization.
\begin{figure}
    \centering
    \includegraphics[width=1\linewidth]{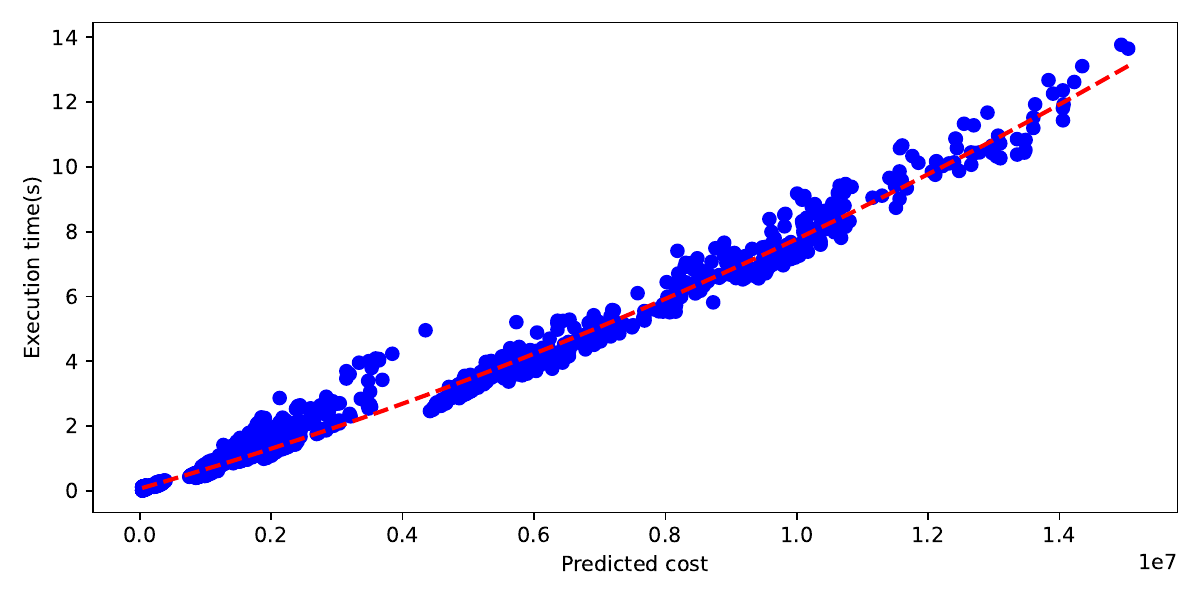}
    \caption{Predicted costs track the actual execution times}
    \label{fig:cost-model}
\end{figure}

\subsection{Impact of Constant Fanout Assumption} 

\begin{figure}
\centering
    \includegraphics[width=0.5\textwidth]{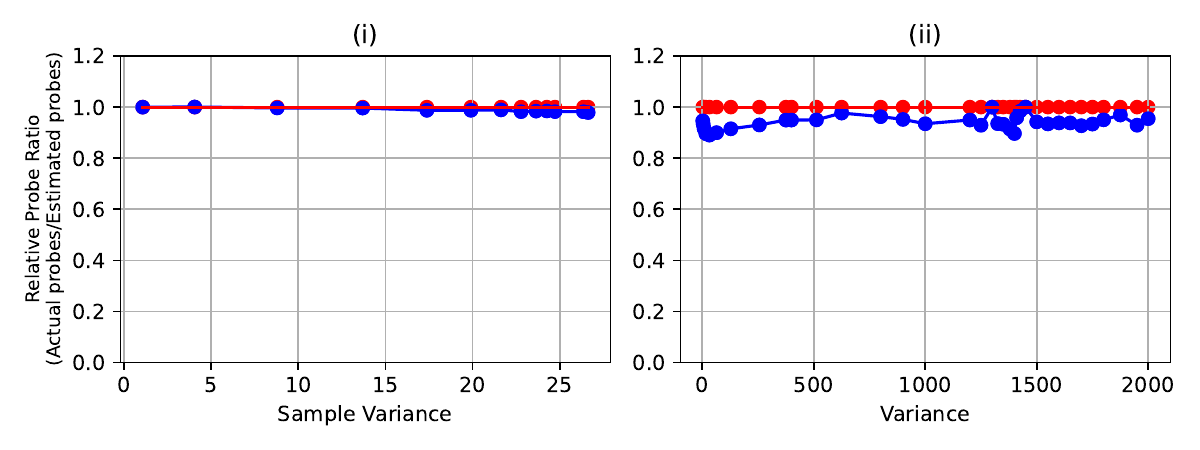}
    \vspace{-15pt}
    \caption{Evaluating the impact of varying fanouts} 
    \vspace{-15pt}
    \label{fig:dist_both}
\end{figure}

For this experiment, we consider the impact of the constant fanout assumption on the estimates. For the 3-2 snowflake query, we generated
a range of datasets with varying skew in the fanout distribution, so that in a single relation, the fanout for a join varied 
across the tuples, distributed either as a {\em truncated} normal distribution ($fo \sim \mathcal{N}(\mu=10, \sigma^2), 1 \le fo \le 2 \mu - 1$) or an exponetial 
distribution. In Figure~\ref{fig:dist_both}, we plot the the estimated as well as the actual number of
probes observed, using the estimated value to normalize, as the population variance increases. 
For the exponential
distribution, the average fanout values were as high as 45.72 for the largest value of variance, indicating a highly skewed
distribution. On the other hand, for the highest value of variance, the normal distribution is almost uniform in the range of
$1$ to $19$. As we can see, 
the estimated number of probes closely match the actual number of probes observed, even when the variance is very high,
demonstrating that our cost model is not very sensitive to the constant fanout assumption.

\subsection{Robustness Evaluation}
Finally, we show results of an experiment to evaluate the robustness of the six approaches, using a varied set of queries in the synthetic and the CE benchmarks. 
For each query, we choose 10 join orders uniformly at random (with the driver relation fixed), and plot the execution times as a box-plot (Figure \ref{fig:robustness}). 
Due to space constraints, we only show this for a subset of the queries, but results were consistent across other datasets and query shapes. 

In the boxplots, we normalize the execution cost for each join order with the largest cost that we see among all the join orders for that method. Hence, this box-plot shows
relative robustness for each method, not absolute robustness.
Across all of the experiments, we consistently see the benefits of using COM for achieving robustness, even for the bitvector-based approach or two-pass full reduction. 
Generally speaking, use of BVP or SJ optimizations achieves higher robustness, and the combination of COM+SJ shows almost no
variations across the join orders (assuming the driver relation is fixed). This is in line with our earlier observation (Theorem \ref{thm:full-reduction}). 

\begin{figure}
    \centering
    \begin{subfigure}{\linewidth}
        \centering
        \includegraphics[scale=0.5]{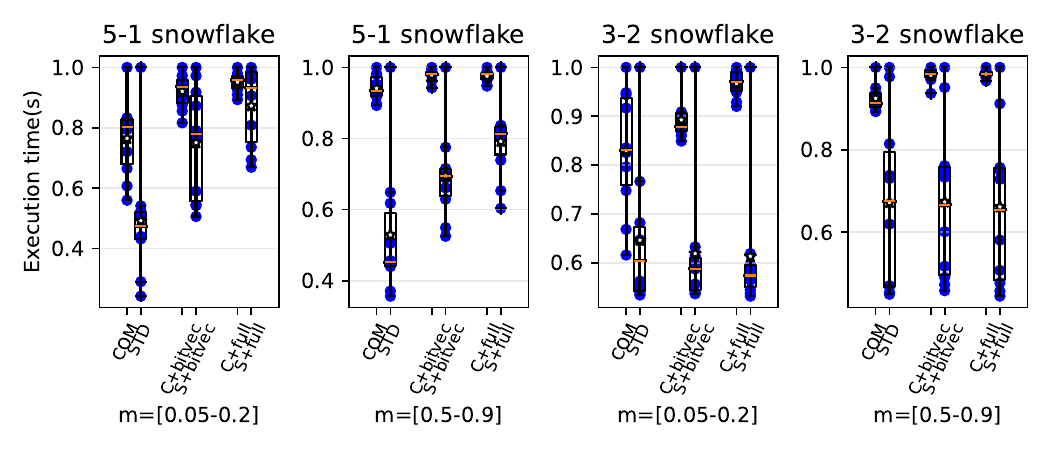}
        \caption{Synthetic Benchmark}
    \end{subfigure}
    \begin{subfigure}{\linewidth}
        \centering
        \includegraphics[scale=0.5]{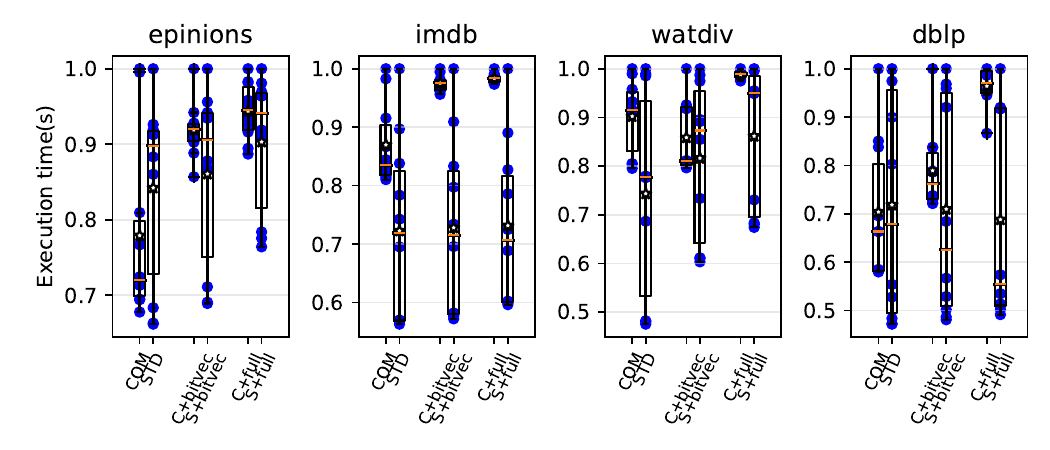}
        \caption{CE Benchmark}
    \end{subfigure}
    \caption{Robustness Evaluation}
        \vspace{0.10in}
    \label{fig:robustness}
\end{figure}

\vspace{10pt}
\section{Related Work} \label{related work}
Here, we cover additional closely related work beyond what was discussed in previous sections.

Because of decades of research, binary joins are utilized in the vast majority of RDBMSs given their simplicity and flexibility to handle a variety of
workloads~\cite{AdoptingWcojs}. 
However, their performance can be sub-optimal, as intermediate results from binary joins often grow larger than the final query result~\cite{hung, kemper, Graefe, eddies_joe,
AdoptingWcojs}. Early techniques like Hash Teams and Eddies addressed this by processing multiple input relations concurrently in a single multi-way join~\cite{Graefe, kemper,
eddies_joe}. But, since they still rely on binary joins, they don't fully mitigate the growth of intermediate results~\cite{AdoptingWcojs}, and cyclic queries can still lead to an
explosion. The recent line of work on worst-case optimal joins presents a new approach to handling cyclic queries through use of new multi-way join algorithms that execute the
join attribute-at-a-time~\cite{hung,Veldhuizen2012LeapfrogTA}. Several recent works have considered combining binary joins and worst-case optimal joins in a single execution
engine~\cite{AdoptingWcojs,remywang}. As with much prior work on query optimization, we restricted our analysis to acyclic queries in this work. In a query engine that continues
to use binary joins, one way to apply our techniques is by following the standard practice of choosing a spanning tree of the join graph (in effect, by ignoring some of the joins during optimization), which will allow us to use the more efficient query optimization heuristics. Another option is to generalize the cost model to handle cyclicity and use the exhaustive algorithm for optimization. For query engines that use WCOJs, there are additional optimization questions (e.g., when to use a WCOJ) that interact with the join order optimization questions. We plan to explore some of these questions in our future work.

With increased prevalence of graph datasets and the need to support analytics over them, subgraph pattern queries are becoming prevalent~\cite{sahu}. 
Since highly connected graph data contains a large number of many-to-many relationships, many-to-many joins are common in such workloads. An explosion in the size of intermediate and output results occurs for even a small number of input tuples when executed with traditional binary join algorithms. Also, long acyclic queries are common in the context of path-finding queries~\cite{lsqb, ldbc}. This has led to much work on multi-way join queries in the context of graph databases~\cite{graphflowdb,kuzu:cidr,amine_continuous}. We are, however, not aware of any work that has considered the query optimization and cost modeling issues that we considered in this paper.

Factorization~\cite{FDB,olt,danolt,wylezo2012cost,zavodny2014a,petrou2015single,graphflowdb} uses 
a combination of vertical (product) and horizontal (union) data partitioning to reduce redundancy in the data while boosting query performance. This method is particularly
effective with schema and queries where all cross products of the answer tuples' projections appear as answer tuples~\cite{AbulBasher2020AnswerGF}. 
However, as we noted earlier, the specific implementations proposed in the prior work are impractical 
for integration into traditional pipelined, vectorized DBMS processors since they use row-based arrays sorted in lexicographic order according to the join order, or rely on representing
input relations as sorted tries, and process and output tries~\cite{kuzu:cidr,graphflowdb}.  
Later work~\cite{wylezo2012cost} on factorized databases considered some cost-based optimization issues, but their focus is on manipulating the factorized
representations within a limited scope, and it is not clear how to apply those techniques to optimizing queries in relational engines. 
A recent system, Answer Graph~\cite{AbulBasher2020AnswerGF}, extends the PostgreSQL query processor to include a join-only subset of SPARQL and conducts a two-step
evaluation of acyclic queries. Similar to the Yannakakis approach, 
the first stage involves a full semi-join reduction, through a series 
of ``forward edge extensions'', followed by cascade deletes called ``node burnbacks''. 
Similarly, Tziavelis et al.~\cite{Tziavelis2019OptimalAF,TziavelisGR21} present techniques for ranked enumeration over conjunctive queries that perform a two-pass full semijoin
reduction to construct a compressed representation in memory for more efficient enumeration. Their work focuses mainly on the optimization problems arising from computing the
outputs in ranked order. 
Another very recent work~\cite{gross2024dynamically} proposes an approach to integrate factorization into DuckDB.
None of these works consider the issues of cost model or join order optimization systematically, and the former two are also not designed to be implemented within a
standard columnar query engine. As discussed earlier, Zhu et al.~\cite{lip} and Ding et al.~\cite{bitvector_aware_query_proc} do consider some of the query 
optimization issues in the context of bitvector-based approach, but do not provide a comprehensive treatment of all different approaches. Finally, a recent work~\cite{diamond} presents an approach that treats the placement of semi-joins or bloomfilters in a 
query plan as an optimization problem, but does not develop complete cost models or optimization algorithms for solving it. Our framework can be naturally generalized to handle 
their plan space.

\vspace{10pt}
\section{Conclusion and Future Work}
In this paper, we formalized and analyzed a core query optimization problem in the context of many-to-many joins. Although a number of different techniques have been proposed to
handle such growing or exploding joins, the query optimization issues have been largely ignored so far. We presented a cost model that more accurately captures the cost of a join
order when using a factorized representation for intermediate results, and presented several optimization algorithms to find the join order given a query and the relevant
statistics. 
Our experimental evaluation, using a modern vectorized implementation of these techniques, clearly demonstrates the need for the more accurate cost modeling, as well as the effectiveness of our techniques. There are number of natural directions for future
work that we are planning to explore, including generalizing to handle cyclic queries; incorporation of WCOJs and deciding when to use them; 
and better selectivity estimation techniques.

\bibliographystyle{ACM-Reference-Format}
\bibliography{sample-base}

\newpage
\appendix

\end{document}